\documentclass[authoryear]{article} 

\RequirePackage{amsthm,amsmath}
\RequirePackage[numbers]{natbib}
\RequirePackage[colorlinks,citecolor=blue,urlcolor=blue]{hyperref}

\usepackage{amssymb,amsfonts,color,xcolor} 
\usepackage{mathtools, thmtools}
\usepackage{bbm,bm}
\usepackage[english]{babel}
\usepackage{authblk}
\usepackage{comment}
\usepackage{graphicx}
\usepackage{float}
\usepackage{subcaption}
\usepackage{enumerate}
\usepackage[title,titletoc,toc]{appendix}
\usepackage{chngcntr}
\usepackage{apptools}
\AtAppendix{\counterwithin{theorem}{section}}

\usepackage[a4paper,bindingoffset=0.5cm,left=2cm,right=2cm,top=2.5cm,bottom=2cm,footskip=.8cm]{geometry} 

\newcommand{\vb}{\vspace{3mm}}

\renewcommand{\bar}{\overline}
\newcommand{\avarbar}{{{\overline{\rm AVaR}_i^{\alpha}}(S,T)}}
\newcommand{\avarbare}{{{\overline{\rm AVaR}_1^{\alpha}}(S,T)}}

\allowdisplaybreaks

\newtheorem{theorem}{Theorem}

\newtheorem{lemma}{Lemma}
\newtheorem{property}{Property}

\theoremstyle{definition}

\newtheorem{remark}{Remark}

\allowdisplaybreaks

\begin{document}

\title{On Capital Allocation for a Risk Measure\\ Derived from Ruin Theory}

\author[1,2]{G.A. Delsing
\footnote{Corresponding author. E-mail address: G.A.Delsing@uva.nl}}
\author[1,3]{M.R.H. Mandjes}
\author[1,4]{P.J.C. Spreij}
\author[1,2]{E.M.M. Winands}

\affil[1]{Korteweg-de Vries Institute 
University of Amsterdam, Science Park 107, 1098 XH Amsterdam, the Netherlands}

\affil[2]{Rabobank, Croeselaan 18, 3521 CB Utrecht, the Netherlands}

\affil[3]{CWI, Science Park 123, 1098 XG Amsterdam, the Netherlands}

\affil[4]{Radboud University, Heyendaalseweg 135, 6525 AJ Nijmegen, the Netherlands}
\maketitle
\begin{abstract} \noindent This paper addresses allocation methodologies for a risk measure inherited from ruin theory. Specifically, we consider a dynamic value-at-risk (VaR) measure defined as the smallest initial capital needed to ensure that the ultimate ruin probability is less than a given threshold. We introduce an intuitively appealing, novel allocation method, with a focus on its application to capital reserves which are determined through the dynamic value-at-risk (VaR) measure. Various desirable properties of the presented approach are derived including a limit result when considering a large time horizon and the comparison with the frequently used gradient allocation method. In passing, we introduce a second allocation method and discuss its relation to the other allocation approaches. A number of examples illustrate the applicability and performance of the allocation approaches.

\vb

\noindent
\end{abstract}

\section{Introduction}\label{sec_intro}
Ruin theory (or risk theory) focuses on analyzing models that describe a company's vulnerability to ruin by studying the riskiness of a firm's reserves. The probability of ruin, i.e., the probability that the capital reserve level of a firm drops below zero, is often used as an insolvency measure. Starting from the seminal works by \cite{Cramer1930} and \cite{Lundberg1903}, a substantial research effort has been spent on determining the ruin probability in a broad range of risk models. In the basic model, the evolution of the capital reserves of a firm over time experiences fluctuations due to losses incurred (amounts claimed in the insurance context) and premiums earned. Initially, the focus of ruin theory has been on the probability of ultimate ruin, i.e., the probability that the capital reserve level {\it ever} drops below zero given the initial capital reserve \(u\). Later these results have been extended in many ways, most notably (i) ruin in finite time, (ii) more general loss/claim arrival processes, and (iii) asymptotics of the ruin probability for \(u\) large. We refer to e.g., \cite{Asmussen2010} for a detailed account.
Whereas most of the existing literature primarily considers the univariate risk setting describing a single capital reserve process, in practice firms often have multiple lines of business. This warrants the study into multivariate risk processes, see e.g., the overview in Chapter XIII.9 of \cite{Asmussen2010}. Various dependence structures have been considered, such as the introduction of common environmental factors, see e.g., \cite{Loisel2007}, or a shared claims process, see e.g., \cite{Picard2003}. 

Ruin theory originating in the actuarial sciences, has become a commonly used tool in the insurance industry. It also has applications in operational risk (see e.g., \cite{Kaishev2008}), credit risk (see e.g., \cite{Chen2009}), and various related fields.

\vb

Traditionally, a firm's risk of insolvency is managed through the control of the initial capital reserve. For instance, companies tune this initial capital reserve level, say $u$, such that their loss within a certain period does not exceed $u$ with a given (low) probability. This risk measure is often referred to as the {\it value-at-risk} (VaR). In many branches of industry, such as insurance and banking, regulation imposes restrictions on the capital reserves: these have to be at least equal to some appropriate VaR over a 1-year horizon. In this paper we work with a risk measure which is defined as the smallest amount of initial capital needed to guarantee a certain probability of solvency over a specified period. As this risk measure can also be seen as the VaR of the maximal aggregate loss encountered over the period, we refer to this as the dynamic VaR measure. This dynamic VaR risk measure, derived from actuarial ruin theory, was first mentioned by \cite{Cheridito2006} and is based on the infinite time ruin probability. Later, in \cite{Trufin2011} \& \cite{Trufin2016}, various properties of this risk measure were examined. 

\vb

As mentioned above, in practice there are good motives to consider the multivariate counterpart of the conventional univariate risk model. Indeed, many firms want their total initial capital to be allocated over multiple business lines. A first reason for this is that it allows them to transfer the cost of holding capital to clients. In the second place, the allocation of expenses across business lines is a necessary activity for financial reporting purposes. Finally, capital allocation provides a useful device for assessing and comparing the performance of the different lines of business including the quantification of risk. 

A variety of capital allocation principles have been proposed in literature. The study of capital allocation can be traced back to the work of \cite{LeMaire1984} discussing capital allocations in a game-theoretic framework. \cite{Cummins2000} provides an overview of several methods for capital allocation in the insurance industry. One of the most important and intensively studied capital allocation methods is the {\it gradient allocation method}, also sometimes referred to as the Euler allocation method, which has been proposed by several authors, see \cite{Tasche2007} for an overview. It is based on the idea of allocating capital according to the infinitesimal marginal impact of each individual risk. In \cite{Tasche1999}, the author argues that allocation based on the gradient principle is the only allocation method that provides the right signals for performance measurement.

\vb

In most of the existing literature on allocation methods the risk is modeled via the terminal value of the risk process at a given time horizon \(T\). A challenge, however, lies in incorporating {\it path-dependent} information in the allocation method (e.g., considering the event of capital reserves dropping below \(0\) before \(T\)), being of interest specifically when focusing on ruin-based risk measures. As mentioned by \cite{Assa2016}, it is particularly difficult to apply existing allocation methods to ruin-based risk measures such as the dynamic VaR measure. Although scarce, ruin-based allocation methods have been proposed in literature, see e.g., \cite{Dhaene2003}, \cite{Frostig2009}, \cite{Li2015} and \cite{Cai2017}. In these works the authors minimize certain multivariate ruin probabilities to obtain the optimal capital allocation across the risk processes or business lines. Another approach was taken by \cite{Assa2016} who reverse-engineer a risk measure with the purpose of addressing the non-trivial problem of capital allocation in a ruin theory context.

\vb

In this paper we first propose a novel capital allocation method when the underlying risk process is of multivariate L\'{e}vy type. This capital allocation is based on the contribution (at the time of ruin, that is) of each of the individual risk processes to the change in the total aggregated capital reserve level. It thus yields an intuitive way of allocating capital taking into account the path-dependent information, and does not require any optimization to be performed. Several other desirable properties of the method will be highlighted. The special case of a multivariate Brownian motion, for which an explicit allocation is found, is dealt with separately. Furthermore, we provide, under certain conditions, an asymptotic result for the allocation in the specific case that ruin over an infinitely long time interval is considered. 

Secondly, we show that for some particular cases (including multivariate scaled Brownian motions with drift) that our new allocation method gives the same capital allocations as the well-known gradient capital allocation method applied to the dynamic VaR measure. To our knowledge, this is the first time in literature that the gradient allocation method is applied to the dynamic version of the VaR measure. In passing we present a second new allocation method, which, when properly defined, is shown to give the same allocations as the first new allocation method when considering an infinite time horizon. The second allocation method is based on the contribution, at the time the supremum of the total aggregated risk process is attained, of each of the individual risk processes to the total aggregated capital reserve level (conditional on the level of the supremum of the total aggregated risk process). We conclude this paper by a series of numerical experiments highlighting some of the differences between the new allocation methods.

\vb

This paper is organized as follows. Section 2 provides a formal model description and some preliminaries including the dynamic VaR measure. Then in Section 3 we propose our novel capital allocation approach, including treatments of (i) the special case of a multivariate Brownian motion, and (ii) the setting in which the time horizon \(T\) is infinite. Section 4 establishes the relation with the gradient capital allocation method and presents an alternative new allocation method.
Numerical examples are provided in Section 5.

\section{Risk model and risk measure}\label{sec_riskmodel}
In this section we introduce our risk model and the risk measure which we focus on throughout this paper. In our setup the risk process is a multidimensional process, whereas the risk measure is defined in terms of the sum of these processes. 
\subsection{Risk model}\label{subsec_riskmodel}
We start by constructing the aggregated risk (or loss) process \(S(\cdot)\). To this end, consider the \(d\)-dimensional real-valued {\it L\'{e}vy process} \(\{\boldsymbol{S}(t),t\geq 0\}\) on the probability space \((\Omega,\mathcal{F},\mathbb{P})\), where \(\boldsymbol{S}(t)=(S_1(t),\ldots,S_d(t))^\top\). We take \(d\geq 2\) and assume \(S_i(0)=0\) for all \(i\).
The process can be characterized by its L\'{e}vy exponent \(\kappa_{\boldsymbol{S}}(\cdot)\), which is given for \(\boldsymbol{\vartheta}:=(\vartheta_1,\ldots,\vartheta_d)^\top\in\mathbb{R}^d\) by
\[\mathbb{E}\left[e^{ \langle \boldsymbol{\vartheta}, \boldsymbol{S}(t)\rangle}\right]=e^{t\kappa_{\boldsymbol{S}}({\boldsymbol \vartheta})};\]
see e.g., \cite{Sato1999} or \cite{Bertoin1996}. The notation \(\langle\cdot,\cdot\rangle\) denotes the usual inner product and the domain of the L\'{e}vy exponents includes imaginary numbers.
The L\'{e}vy exponent of the process \(\boldsymbol{S}(\cdot)\) is necessarily of the form
\[\kappa_{\boldsymbol{S}}(\boldsymbol{\vartheta})=\langle \boldsymbol{c} ,\boldsymbol{\vartheta}\rangle+\frac{1}{2}\langle \boldsymbol{\vartheta},\boldsymbol{\Sigma\vartheta}\rangle +\int_{\mathbb{R}^d}\left(e^{\langle\boldsymbol{\vartheta},\boldsymbol{y}\rangle}-1-\langle\boldsymbol{\vartheta},\boldsymbol{ y}\rangle \mathbbm{1}_{D(\boldsymbol{y})}\right)\Pi({\rm d}\boldsymbol{y}),\]
where \(\boldsymbol{c}:=(c_1,\ldots,c_d)^\top\in\mathbb{R}^d\), \(\boldsymbol{\Sigma}:=(\Sigma_{ij}))_{i,j\in \{1,\ldots,d}\) a symmetric non-negative definite matrix on \(\mathbb{R}^d\) and \(D:=\{\boldsymbol{y}:|\boldsymbol{y}|\leq 1\}\) the closed unit ball. One refers to \((\boldsymbol{c},\boldsymbol{\Sigma},\Pi)\) as the characteristic triplet. The first term corresponds to a deterministic drift, the second term to a diffusion part, and the third term to the process' 
jumps. Regarding this third part, \(\Pi(\cdot)\) is often referred to as the L\'{e}vy measure on \(\mathbb{R}^d-\{0\}\), and satisfies \(\Pi((0,\ldots,0))=0\) and \(\int_{\mathbb{R}^d}\left(|\boldsymbol{y}|^2\wedge 1\right)\Pi({\rm d}\boldsymbol{y})<\infty\).
{In this paper we often introduce assumptions of the form \(\mathbb{E}\left[\big|\boldsymbol{S}(t)\big|\right]<\infty\) for all \(t\). As a consequence of Theorem 25.3 in \cite{Sato1999} (as highlighted in Example 25.12 in \cite{Sato1999}) this condition holds true if and only if \(\int_{|\boldsymbol{x}|>1}|\boldsymbol{x}|\Pi({\rm d}\boldsymbol{x})<\infty\). As a result, the condition \(\mathbb{E}\left[\big|\boldsymbol{S}(t)\big|\right]<\infty\) for all \(t\) is also equivalent to \(\mathbb{E}\left[\big|\boldsymbol{S}(1)\big|\right]<\infty\).}

We define the aggregated risk process by
\[S(\cdot):=\sum_{i=1}^dS_i(\cdot),\]
which by Proposition 11.10 of \cite{Sato1999} is again a L\'{e}vy process, whose L\'{e}vy exponent is denoted by \(\kappa(\cdot)\). In a ruin context, we let \(u-S(t)\) represent the capital surplus at time \(t\) of the entire firm (defined as the sum of its separate lines of business), given that the initial capital reserve level was \(u>0\). This means that the probability of ruin over some time horizon can be expressed as the probability that the risk/loss process \(S(\cdot)\) exceeds the level \(u\) at some point over the time horizon, i.e.,
\[\psi(u,T):=\mathbb{P}\left(\tau(u)\leq T\right)=\mathbb{P}\left(\sup_{t\in[0,T]}S(t)\geq u\right), \ \ \ \ \psi(u,\infty):=\mathbb{P}\left(\tau(u)<\infty\right)=\mathbb{P}\left(\sup_{t\in[0,\infty)}S(t)\geq u\right),\]
where the time of ruin is then defined as
\[\tau(u):=\inf\{t\geq 0:S(t)\geq u\}.\]

In this paper we focus on techniques pertaining to the determination of an appropriate initial capital reserve level \(u\) for the aggregated risk process \(S(\cdot)\) by imposing a bound on the probability of ruin, and subsequently allocating the capital $u$ over the individual risk processes. 

\subsection{Risk Measure Derived from Ruin Theory}\label{subsec_riskmeasure}
In this subsection we present the {\it dynamic value-at-risk} (VaR) measure for the aggregated risk process \(\{S(t)\}_{t\in[0,T)}\) to determine capital reserves. This risk measure has been derived from ruin theory and was introduced by \cite{Trufin2011}. The dynamic VaR measure is defined as the minimum initial capital reserve level \(u\) such that the probability of ruin is below a given threshold. In other words, for (typically small) \(\alpha\in[0,1]\), aggregated risk process $S(\cdot)$, and time horizon $T>0$,
\begin{equation}\label{eq_VaR}
{\rm VaR}^{\alpha}(S,T):=\inf\{u\geq 0\,|\, \psi(u,T)\leq \alpha\}.
\end{equation}
In the remainder of this paper special attention will be paid to the allocation of capital when the capital reserve level has been determined by this dynamic VaR measure.
 
\vb 

Desirable properties of risk measures have been extensively analyzed in literature. For a more extensive account of risk measures and their properties we refer to, e.g., the original works by \citet*{Artzner1997, Artzner1999}. These works coin the concept of a {\it coherent} risk measure by introducing a list of four \textit{axioms}.
In the context of the infinite time ruin probability, \({\rm VaR}^{\alpha}(S,\infty)\) and its properties have been studied by \cite{Trufin2011}. In line with those results, the dynamic VaR measure with a finite time horizon satisfies the following axioms:
\begin{description}
\item[Axiom 1] (Translation Invariance) For \(\gamma \in \mathbb{R}\), \({\rm VaR}^{\alpha}(S+\gamma,T)={\rm VaR}^{\alpha}(S,T)+\gamma\).
\item[Axiom 2] (Positive Homogeneity) For every \(\gamma>0\), \({\rm VaR}^{\alpha}(\gamma S,T)=\gamma{\rm VaR}^{\alpha}(S,T)\).
\item[Axiom 3] (Monotonicity) For \(S_1(t)\leq S_2(t)\) a.s. for all \(t\in[0,T)\) then \({\rm VaR}^{\alpha}(S_1,T)\leq {\rm VaR}^{\alpha}(S_2,T)\).
\end{description}

Note that in this context \(S(\cdot)\) is referred to as a `loss'. The remaining axiom (generally not satisfied by VaR-type risk measures) for a coherent risk measure is sub-additivity, i.e.,

\begin{description}
\item[Axiom 4] (Sub Additivity) \({\rm VaR}^{\alpha}(S_1+S_2,T)\leq {\rm VaR}^{\alpha}(S_1,T) +{\rm VaR}^{\alpha}(S_2,T)\).
\end{description}

In case the risk processes are scaled Brownian motions with drift, the risk measure \({\rm VaR}^{\alpha}(S,\infty)\) is also sub-additive. This is explicitly shown in Section~\ref{subsec_BM_GVAR}.

\section{A capital allocation approach}\label{sec_Aallocation}
In this section we focus on the allocation of capital of the aggregated risk process \(S(\cdot)\) over the individual risk processes \(S_i(\cdot)\). More specifically, we introduce an intuitively appealing novel method to allocate the (initial) capital reserve level of the aggregated risk process \(u\) to capital reserves for the individual risk processes \(u_i\) based on the risk contribution of the processes at the time of ruin. Special focus is given to its application to capital reserves which are determined through the dynamic VaR measure. In addition to presenting some nice properties of the proposed allocation method, in particular over an infinite time horizon, we also find an explicit expression for the special case of scaled Brownian motions with drift.

\subsection{The capital allocation method and its properties}\label{subsec_AVaR}
As mentioned in Section~\ref{subsec_riskmeasure}, an intuitive way to determine the (initial) capital level \(u\) for the aggregated risk process \(S(\cdot)=\sum_{i=1}^d S_i(\cdot)\) is by use of the dynamic VaR measure. In other words, by determining the minimum reserve level needed to ensure that the probability of ruin is below some threshold \(\alpha\), i.e., \(\mathbb{P}\left(S(\tau(u))\geq u\right)\leq \alpha\). A natural way to determine the contribution of risk process \(S_i(\cdot)\) (or business line \(i\)) to \(u\), is to consider its contribution to \(S(\tau(u))\). Given the capital reserve level \(u\) for the aggregated process \(\{S(t)\}_{t\in [0,T]}\), we propose to allocation \(K_i(u,S,T)\) to the \(i\)\textsuperscript{th} risk process \(S_i(\cdot)\), where
\begin{equation}\label{eq_Kallocation}
K_i(u,S,T):=c_i(u,S,T)\,u, \ \ \ c_i(u,S,T):=\frac{\mathbb{E}[S_i(\tau(u))\,|\,\tau(u)\leq T]}{\mathbb{E}[S(\tau(u))\,|\,\tau(u)\leq T]}.
\end{equation}
For the infinite time horizon we consider the allocation \(K_i(u,S,\infty):=c_i(u,S,\infty)\,u\), with \[c_i(u,S,\infty):=\frac{\mathbb{E}[S_i(\tau(u))|\tau(u)<\infty]}{\mathbb{E}[S(\tau(u))|\tau(u)< \infty]}.\] The \(c_i(u,S,T)\) add up to 1 (but are not necessarily positive), entailing that for \(u_i:=K_i(u,S,T)\) we have \(\sum_{i=1}^du_i=u\). This is known as the {\it full allocation principle}, which evidently is a desirable property for an allocation method. More properties of this allocation method are discussed below. To ensure that the allocation method is properly defined we assume from this point onward that \(\mathbb{P}(\tau(u)\leq T)>0\).

\begin{remark}\label{rem_AVAR}
When the capital level \(u\) is determined using the dynamic VaR measure presented in Equation \eqref{eq_VaR}, the proposed capital allocation method gives, for risk process \(i\),
\[{\rm AVaR}_i^{\alpha}(S,T):=K_i({\rm VaR}^{\alpha}(S,T),S,T).\]
\end{remark}

Various desirable properties of allocation methods can be found in literature, see for example, \cite{Denault2001} for an introduction into coherent allocation principles. We now present some properties of the allocation \(K_i(u,S,T)\). 

\begin{property}\label{property1}
The allocation \(K_i(u,S,T)\) possesses the following properties:
\begin{description}
\item[(i)] The allocated risk measure \(K_i(u,S,T)\) is positively homogeneous, that is, \(K_i\left(u,\gamma S_i,T\right)=\gamma K_i(u,S,T)\) for any constant \(\gamma>0\).
\item[(ii)] The allocation satisfies the full allocation principle, that is, \(\sum_{i=1}^d K_i(u,S,T)=u\).
\item[(iii)] If \(\mathbb{E}[S(\tau(u))|\tau(u)< T]=u\), then the allocated risk is deterministic when the marginal risk is deterministic, i.e., \(K_i(u,S,T)=s_i\) if \(S_i(\cdot)\equiv s_i\) is deterministic.
\end{description}
\end{property}
\begin{proof}
All properties follow directly from definition \eqref{eq_Kallocation}.
\end{proof}

\subsection{Infinite horizon}\label{subsec_asymptotics}
In this subsection we are interested in the properties of the allocation method as given by \eqref{eq_Kallocation} when considering an infinite time horizon, i.e., the case \(T= \infty\). To avoid trivialities, we assume throughout that the process \(S(\cdot)\) is a L\'{e}vy process with negative drift, i.e., \(\kappa'(0)<0\) and that it is not the negative of a subordinator (in which case \(\psi(u,\infty)=0\) for all \(u>0\)). Suppose we are in the light-tailed regime, in the sense that the equation \(\kappa(\vartheta)=0\) has a real positive solution. This root, typically referred to as the `Cram\'{e}r root', we denote by \(\vartheta^*\in(0,\infty)\).

We apply now an exponential change of measure. Following \cite[Section~3.3]{Kyprianou2006} we introduce the alternative measure \(\mathbb{Q}\) defined on \(\mathcal{F}:=\sigma(S(t),t\geq 0)\). The restrictions of $\mathbb{Q}$ and $\mathbb{P}$ to the $\sigma$-algebras $\mathcal{F}_t:= \sigma(S(u),u\geq t)$, conveniently denoted $\mathbb{Q}_t$ and $\mathbb{P}_t$, are assumed to be mutually absolutely continuous with likelihood ratio process $Z=\{Z_t,\,t\geq 0\}$, a martingale under $\mathbb{P}$, taking the form (recall $\kappa(\vartheta^*)=0$)
\[
Z_t=e^{\vartheta^* S(t)-\kappa(\vartheta^*)t}=e^{\vartheta^* S(t)}.
\]
For a stopping time $\tau$, we will use that the measures $\mathbb{Q}$ and $\mathbb{P}$ restricted to the $\sigma$-algebra $\mathcal{F}_\tau$, these restrictions denoted $\mathbb{Q}_\tau$ and $\mathbb{P}_\tau$, are also mutually absolutely continuous, with likelihood ratio $Z_\tau$ on the set $\{\tau<\infty\}$, see~\cite[Section~III.3]{jacodshiryaev} for further details on measure transformations. Below we apply this to $\tau=\tau(u)$.

Furthermore, with \({\boldsymbol S}(t):=(S_1(t),\ldots, S_d(t))^\top\) having a density function \(f_{{\boldsymbol S}(t)}^{\mathbb{P}}({\boldsymbol x})\) under $\mathbb{P}$ for \({\boldsymbol x}=(x_1,\ldots x_d)^\top\in\mathbb{R}^d\), under \(\mathbb{Q}\) the density becomes
\[f_{{\boldsymbol S}(t)}^{\mathbb{Q}}({\boldsymbol x})=f_{{\boldsymbol S}(t)}^{\mathbb{P}}({\boldsymbol x})\,e^{\vartheta^*\sum_i x_i}.\]
For a more comprehensive overview of this exponential change of measure, often referred to as exponential tilting, we refer the reader to section 3.3 in \cite{Kyprianou2006}.

\vb

Under the exponential change of measure, the process \(\boldsymbol{S}(\cdot)\) is still a L\'{e}vy process (see Theorem 3.9 in \cite{Kyprianou2006}). Concretely, under the alternative measure \(\mathbb{Q}\) the L\'{e}vy exponent of the multivariate process \(\boldsymbol{S}\) is given by \[\kappa_{\boldsymbol{S}}^{\mathbb{Q}}(\boldsymbol{\vartheta}):=\kappa_{\boldsymbol{S}}(\boldsymbol{\vartheta}+\vartheta^*{\boldsymbol 1})-\kappa_{\boldsymbol{S}}(\vartheta^*{\boldsymbol 1}),\] where \(\boldsymbol 1\) denotes the unit vector of dimension \(d\). Similarly, \(\kappa^{\mathbb{Q}}(\vartheta)=\kappa(\vartheta+\vartheta^*)-\kappa(\vartheta^*)\) denotes the L\'{e}vy exponents of the process \(S(\cdot)\) under the new measure. As a result, we find \(\mathbb{E}^{\mathbb{Q}}[S(1)]=(\kappa^{\mathbb{Q}}) '(0)=\kappa'(\vartheta^*)>0\) (which follows from the fact that \(\kappa(\vartheta)\rightarrow\infty\) as \(\vartheta\rightarrow\infty\), the convexity of \(\kappa(\cdot)\) and \(\kappa(0)=0\)), which gives \(\mathbb{Q}(\tau(u)<\infty)=1\).
We find,
\begin{equation}\label{eq_mi}
\mathbb{E}^{\mathbb{Q}}[S_i(t)]=\mathbb{E}\left[S_i(t)e^{\vartheta^* S(t)}\right]=\frac{\partial}{\partial \vartheta_i}\mathbb{E}\left[e^{\langle{\boldsymbol\vartheta},{\boldsymbol S}(t)\rangle}\right]\Bigg|_{{\boldsymbol\vartheta}=\vartheta^*{\boldsymbol 1}}=t\,m_i,
\end{equation}
with
\[m:=\sum_{i=1}^d m_i, \ \ \ \ \ m_i:=\frac{\partial}{\partial \vartheta_i}\mathbb{E}\left[e^{\langle{\boldsymbol\vartheta},{\boldsymbol S}(1)\rangle}\right]\Bigg|_{{\boldsymbol\vartheta}=\vartheta^*{\boldsymbol 1}}=e^{\kappa_{\boldsymbol{S}}(\vartheta^*{\boldsymbol 1})}\frac{\partial\kappa_{\boldsymbol{S}}(\boldsymbol{\vartheta})}{\partial \vartheta_i}\Bigg|_{{\boldsymbol\vartheta}=\vartheta^*{\boldsymbol 1}}.\]
Here it has been used that, in case it exists, the mean of the position of a L\'{e}vy process at time \(t\) is linear in \(t\) (and that the exponentially twisted version of a L\'{e}vy process is again a L\'{e}vy process).

\begin{remark}
Note that in case there is a positive drift, i.e., \(\mathbb{E}[S(1)]=\kappa'(0)>0\), then we find that \(\mathbb{P}(\tau(u)<\infty)=1\) and a change of measure is not necessary to ensure that the time of ruin is finite. 
\end{remark}

We first present a Lemma that is needed in the proof of the main result of this section, Theorem~\ref{th_asymptotic}.

\begin{lemma}\label{lem_UI}
Suppose that \(\mathbb{E}^{\mathbb{Q}}\left[\left|S_i(1)\right|^p\right]<\infty\) for all \(i\) and \(p\geq 1\). Then under the probability measure \(\mathbb{Q}\), for \(i=1,\ldots,d\),
\begin{enumerate}[(i)]
\item \[\frac{S_i(\tau(u))}{u}-\frac{m_i}{m}\xrightarrow{a.s.} 0 \ as \ u\rightarrow\infty;\]
\item \[\left\{\left|\frac{S_i(\tau(u))}{u}-\frac{m_i}{m}\right|^p, u\geq 1\right\} \mbox{ is  uniformly  integrable};\]
\item \[\mathbb{E}^{\mathbb{Q}}\left|\frac{S_i(\tau(u))}{u}-\frac{m_i}{m}\right|^p\rightarrow 0 \ as \ u\rightarrow\infty.\]
\end{enumerate}
\end{lemma}
\begin{proof}
(i) Follows directly from Theorem 5.1 from \cite{Gut1996}.

(ii) We use a similar approach as in \cite{Gut1975}. Note that by the definition of a L\'{e}vy process for integer \(n\geq 1\), \(S(n)=\sum_{k=1}^nS(k)-S(k-1)\) is a sum of i.i.d. random variables. For the finite stopping time \(\tau(u)\) under the filtration \(\mathcal{F}^{\mathbb{Q}}\) we define the positive, integer valued stopping time \(\tau'(u)\) (under the same filtration) as follows:
\begin{equation}
\tau'(u)=1 \ \text{for} \ 0\leq\tau(u)\leq 1 \ \ \ \text{and} \ \ \ \tau'(u)=n \ \text{for} \ n-1<\tau(u)\leq n.
\end{equation}
The integer valued stopping time \(\tau'(u)\) can be seen as the stopping of the random walk \(S(n)\) with i.i.d. random variables.
Next, let 
\[U_i(n):=\sup_{n-1\leq s\leq n}|S_i(s)-S_i(n-1)|.\] Then \(\{U_i(n),n\geq 1\}\) is a i.i.d. sequence of random variables with finite expectation under \(\mathbb{Q}\) due to \(\mathbb{E}^{\mathbb{Q}}\left|S_i(1)\right|<\infty\), i.e.,\
\begin{equation}\label{eq_Usup}
\mathbb{E}^{\mathbb{Q}}\left[U_i(n)^p\right]=\mathbb{E}^{\mathbb{Q}}\left[U_i(1)^p\right]=\mathbb{E}^{\mathbb{Q}}\left[\sup_{0\leq s\leq 1}\left|S_i(s)\right|^p\right]<\infty.
\end{equation}
The last inequality follows from Lemma 2.3 of \cite{Gut1975} (see also Section 3 in \cite{Gut1975} for a similar application).
Furthermore, 
\[\left|S_i(\tau(u))-S_i(\tau'(u))\right|\leq 2U_i(\tau'(u)).\]
For all \(i=1,\ldots,d\) we then get for \(u>0\),
\begin{align}
\left|\frac{S_i(\tau(u))}{u}\right|^p&\leq 2^{p-1}\left|\frac{S_i(\tau(u))-S_i(\tau'(u))}{u}\right|^p+2^{p-1}\left|\frac{S_i(\tau'(u))}{u}\right|^p\nonumber\\
&\leq 2^{p-1}\left(\frac{2U_i(\tau'(u))}{u}\right)^p+2^{p-1}\left|\frac{S_i(\tau'(u))}{u}\right|^p\label{eq_Sineq},
\end{align}
where the first inequality follows from a combination of triangle inequality and Jenssen's inequality.

As a result of Theorem 3.7.1 of \cite{Gut2009} we know that
\[\left\{\left(\frac{\tau'(u)}{u}\right)^p,u\geq 1\right\} \mbox{ is  uniformly  integrable.}\]
By Theorem 1.6.1 of \cite{Gut2009} we then get that \(\left\{ \left|\frac{S_i(\tau'(u))}{u}\right|^p,u\geq 1\right\} \) is uniformly integrable and

 \( \left\{ \left(\frac{U_i(\tau'(u))}{u}\right)^p,u\geq 1\right\} \) is uniformly integrable using \eqref{eq_Usup}. As a result of Lemma A.1.3 of \cite{Gut2009} together with \eqref{eq_Sineq} we get that \(\left\{ \left|\frac{S_i(\tau(u))}{u}\right|^p,u\geq 1\right\} \) is uniformly integrable and the final result follows.

(iii) This follows from Theorem A.1.1 of \cite{Gut2009} together using the results of (i) and (ii).
\end{proof}

\begin{theorem}\label{th_asymptotic}
Let $\mathbb{E}^{\mathbb{Q}}\left[|S_i(1)|\right]<\infty$ for all $i=1,\ldots,m$. Then the limiting allocated proportion of capital is given by
\begin{equation}
\lim_{u\rightarrow\infty}c^i(u,S,\infty)=\frac{m_i}{m}, \ \ \ {\rm with} \ m:=\sum_{i=1}^d m_i \ {\rm and} \ m_i:=\frac{\partial}{\partial \vartheta_i}\mathbb{E}\left[e^{\langle{\boldsymbol\vartheta},S(1)\rangle}\right]\Bigg|_{{\boldsymbol\vartheta}=\vartheta^*{\boldsymbol 1}}
\end{equation}
\end{theorem}
\begin{proof}
Using expression \eqref{eq_mi} and Lemma 2.3 \cite{Gut1996}, which gives $\mathbb{E}[S(\tau(u))|\tau(u)<\infty]=u+o(u)$ as $u\rightarrow\infty$, we need to prove that
\begin{align}\label{eq_th_toprove}
\lim_{u\rightarrow\infty}\frac{\mathbb{E}[S_i(\tau(u))|\tau(u)<\infty]}{u}= \frac{\mathbb{E}^{\mathbb{Q}}[S_i(1)]}{\mathbb{E}^{\mathbb{Q}}[S(1)]}.
\end{align}

First we note that
\begin{align}\label{eq_PQ}
\frac{\mathbb{E}[S_i(\tau(u))|\tau(u)<\infty]}{u}&=\frac{\mathbb{E}^{\mathbb{Q}}\left[S_i(\tau(u))e^{-\vartheta^*S(\tau(u))}\right]}{u\psi(u,\infty)}\\
&=\mathbb{E}^{\mathbb{Q}}\left[\left(\frac{S_i(\tau(u))}{u}-\frac{m_i}{m}\right)\frac{e^{-\vartheta^*S(\tau(u))}}{\psi(u,\infty)}\right]+\frac{m_i}{m}.\label{eq:overstap}
\end{align}
Here it should be borne in mind that $\mathbb{Q}(\tau(u)<\infty)=1$ 
and that consequently the likelihood ratio $\frac{d\mathbb{P}_\tau}{d\mathbb{Q}_\tau}$ on $\mathcal{F}_{\tau(u)}$ 
equals $Z_{\tau(u)}^{-1}=e^{-\vartheta^*S(\tau(u))}$, $\mathbb{Q}$-a.s.
The  equality \eqref{eq:overstap} follows from the fact that $\mathbb{E}^{\mathbb{Q}}\left[\frac{e^{-\vartheta^*S(\tau(u))}}{\psi(u,\infty)}\right]= 1$.

We are left to prove 
\begin{equation}\label{eq_meanconvergence}
\lim_{u\rightarrow\infty}\mathbb{E}^{\mathbb{Q}}\left|\left(\frac{S_i(\tau(u))}{u}-\frac{m_i}{m}\right)\frac{e^{-\vartheta^*S(\tau(u))}}{\psi(u,\infty)}\right|=0.
\end{equation}

due to the 'Cram\'{e}r-Lundberg' result for L\'{e}vy process (see \cite{Bertoin1994}) there is a positive constant $C$ such that $\psi(u,\infty)e^{\vartheta^*u}\rightarrow C$ as $u\rightarrow\infty$; actually, $C=\lim_{u\rightarrow\infty}e^{\vartheta^*u}\mathbb{E}^{\mathbb{Q}}\left[e^{-\vartheta^*S(\tau(u))}\right]$. Furthermore, $S(\tau(u))\geq u\geq 0$ gives $0\leq \frac{e^{-\vartheta^*S(\tau(u))}}{\psi(u,\infty)}\leq\frac{e^{-\vartheta^*u}}{\psi(u,\infty)}<\infty$ and implies that $C\in(0,1]$. Notice that for $\delta\in(0,C)$ and $u$ sufficiently large,

\begin{align*}
\mathbb{E}^{\mathbb{Q}}\left|\left(\frac{S_i(\tau(u))}{u}-\frac{m_i}{m}\right)\frac{e^{-\vartheta^*S(\tau(u))}}{\psi(u,\infty)}\right|&\leq\frac{e^{-\vartheta^*u}}{\psi(u,\infty)}\mathbb{E}^{\mathbb{Q}}\left|\frac{S_i(\tau(u))}{u}-\frac{m_i}{m}\right|\\
&\leq\frac{e^{-\vartheta^*u}}{(C-\delta)e^{-\vartheta^*u}}\mathbb{E}^{\mathbb{Q}}\left|\frac{S_i(\tau(u))}{u}-\frac{m_i}{m}\right|=\frac{1}{(C-\delta)}\mathbb{E}^{\mathbb{Q}}\left|\frac{S_i(\tau(u))}{u}-\frac{m_i}{m}\right|
\end{align*}

We obtain by Lemma~\ref{lem_UI}, for $u\rightarrow\infty$ 
 $$\mathbb{E}^{\mathbb{Q}}\left|\frac{S_i(\tau(u))}{u}-\frac{m_i}{m}\right|\rightarrow 0.$$
This proves the result.
\end{proof}

In the special case that the aggregated process \(S(\cdot)\) is a spectrally negative L\'{e}vy process, i.e.\ a process which does not contain positive jumps and the L\'{e}vy measure has support on \((-\infty,0]\) only, Theorem~\ref{th_asymptotic} holds true for all \(u>0\). This statement is formalized in Theorem~\ref{th_wald}. It is a result of the fact that the lack of positive jumps of the spectrally negative L\'{e}vy process allows us to write \(S(\tau(u))=u\) on \(\{\tau(u)<\infty\}\) (see Section 8.1 in \cite{Kyprianou2006}). For a more extensive overview of spectrally negative L\'{e}vy processes, we refer the reader to Chapter VII in \cite{Bertoin1996} or Chapter 8 in \cite{Kyprianou2006}. 

\begin{theorem}\label{th_wald}
For a spectrally negative L\'{e}vy process \(S(\cdot)\), suppose that for every \(t\geq 0\) and for all \(i=1,\ldots,d\), \(\mathbb{E}^{\mathbb{Q}}\left[\big|S_i(t)\big|\right]<\infty\). Then, if $m>0$ and any $u>0$,
\begin{equation}\label{eq_th_wald}
c_i(u,S,\infty)=\frac{m_i}{m}.
\end{equation}
\end{theorem}
\begin{proof}
By the definition of $c_i(u,S,\infty)$ and the relation between ${\mathbb P}$ and ${\mathbb Q}$, 
\begin{equation}\label{eq_th_inf_time1}
c_i(u,S,\infty)=\frac{\mathbb{E}[S_i(\tau(u))\,|\,\tau(u)<\infty]}{\mathbb{E}[S(\tau(u))\,|\,\tau(u)<\infty]}=\frac{\mathbb{E}^{\mathbb{Q}}\left[S_i(\tau(u))e^{-\vartheta^* S(\tau(u))}\right]}{\mathbb{E}^{\mathbb{Q}}\left[S(\tau(u))e^{-\vartheta^* S(\tau(u))}\right]}=\frac{\mathbb{E}^{\mathbb{Q}}\left[S_i(\tau(u))\right]}{\mathbb{E}^{\mathbb{Q}}\left[S(\tau(u))\right]},
\end{equation}
where the last equality follows from the spectral negativity of \(S(\cdot)\) which allows us to write \(S(\tau(u))=u\) on \(\{\tau(u)<\infty\}\) (see Section 8.1 in \cite{Kyprianou2006}).
As a result of Theorem 3.1 of \cite{Gut1975} and by the assumption that \(\mathbb{E}^{\mathbb{Q}}[|S(1)|]<\infty\) we have \(\mathbb{E}^{\mathbb{Q}}[\tau(u)]<\infty\). 
By Wald's equation in continuous time, as stated on page 380 of \cite{Doob1990}, we thus have that \(\mathbb{E}^{\mathbb{Q}}[S_i(\tau(u))]\) and \(\mathbb{E}^{\mathbb{Q}}[S(\tau(u))]\) exist and are given by 
\[\mathbb{E}^{\mathbb{Q}}[S_i(\tau(u))]=\mathbb{E}^{\mathbb{Q}}[S_i(1)]\,\mathbb{E}^{\mathbb{Q}}[\tau(u)]<\infty, \ \ \ \ \mathbb{E}^{\mathbb{Q}}[S(\tau(u))]=\mathbb{E}^{\mathbb{Q}}[S(1)]\,\mathbb{E}^{\mathbb{Q}}[\tau(u)]<\infty.\]
Substituting this into Equation \eqref{eq_th_inf_time1} gives the final result. Note that this is properly defined due to the fact that \(\mathbb{E}^{\mathbb{Q}}[S(1)]=m>0\) as derived before.
\end{proof}

\begin{remark}\label{rem_Hall}
In Corollary 1 of \cite{Hall1970} more general conditions are given to guarantee the existence of \(\mathbb{E}^{\mathbb{Q}}\left[S_i(\tau(u))\right]\) and \(\mathbb{E}^{\mathbb{Q}}\left[S(\tau(u))\right]\). We will now briefly mention them. If either (i) there exists a real \(\vartheta_i>0\) such that \(\mathbb{E}^{\mathbb{Q}}\left[e^{-\vartheta_i S_i(1)}\right]\leq 1\), or (ii) \(\mathbb{E}^{\mathbb{Q}}\left[S_i(1)\right]=0\), then \(\mathbb{E}^{\mathbb{Q}}\left[S_i(\tau(u))\right]\) exists and is given by \(\mathbb{E}^{\mathbb{Q}}[S_i(\tau(u))]=\mathbb{E}^{\mathbb{Q}}[S_i(1)]\,\mathbb{E}^{\mathbb{Q}}[\tau(u)]\). To derive the latter, we have made use of the fact that \(\lim_{t\rightarrow\infty}{S_i(t)}/{t}=\mathbb{E}^{\mathbb{Q}}[S_i(1)]\) almost surely under the \(\mathbb{Q}\)-measure by Theorem 36.5 of \cite{Sato1999}. With \(\kappa^{\mathbb{Q}}_i(\cdot)\) denoting the L\'{e}vy exponent of the process \(S_i(\cdot)\), condition (i) above is equivalent to the existence of a real \(\vartheta_i<0\) such that \(\kappa^{\mathbb{Q}}_i(\vartheta_i)\leq 0\). The same reasoning also applies to \(\mathbb{E}^{\mathbb{Q}}\left[S(\tau(u))\right]\).
\end{remark}

Under the conditions of Theorems \ref{th_asymptotic} \& \ref{th_wald}, the allocation fraction over an infinite time horizon does not depend on the capital level \(u\). This is an attractive property as it reflects robustness of the allocation with respect to \(u\).

\subsection{Special case: Brownian motion}\label{subsec_BM_AVAR}
In this subsection we simplify the risk model of Section~\ref{sec_riskmodel} by considering only the diffusion part of the L\'{e}vy process. We assume that the processes \(S_i(\cdot)\) are scaled Brownian motions with drift. Due to the fact that these processes are continuous and obey various convenient properties, the proposed capital allocation method \eqref{eq_Kallocation} can be made explicit.

Concretely, the risk processes in this subsection are assumed to be of the form \(S_i(\cdot):=r_it+B_i(t)\), where \(\boldsymbol{B}(t):=(B_1(t),\ldots,B_d(t))^\top\) is a \(d\)-dimensional Brownian motion with zero drift and covariance matrix \[\mathbb{E}[\boldsymbol{B}(t)\boldsymbol{B}(t)^{\top}]=\boldsymbol{\Sigma}\, t,\] such that \(\boldsymbol{\Sigma}:=(\sigma_{i,j}^2)_{i,j\in\{1,\ldots,d\}}\) with \(\sigma_{i,j}^2=\sigma_{j,i}^2=\rho_{i,j}\sigma_{i,i}\sigma_{j,j}\) for correlations \(\rho_{i,j}\in[-1,1]\). In this case \(S(t):=\sum_{i=1}^dS_i(t)\) is also a Brownian motion with drift \(r:=\sum_{i=1}^d r_i\) and variance coefficient \(\sigma^2:=\sum_{i,j=1}^d\sigma_{i,j}^2\) (assumed to be strictly positive, to rule out trivial cases). Note that this is a spectrally negative L\'{e}vy process and due to the continuous sample paths of Brownian processes, there is no overshoot at the first passage time, i.e., \(S\left(\tau(u)\right)=u\), and its supremum equals its maximum. As a result, Theorem~\ref{th_wald} can be applied to scaled Brownian motions with drift.

\vb 

In the next theorem we present an explicit expression for the proposed allocation method \eqref{eq_Kallocation}. It considers the case of a finite time horizon. Below, \(\Phi(\cdot)\) denotes the cumulative distribution function of a standard normal random variable.

\begin{theorem}\label{th_BM_alloc}
For the multivariate scaled Brownian motion with drift the allocated initial reserve of component \(i \in \{1,\ldots,d\}\) is given by
\begin{equation}\label{eq_th_BM_alloc}
K_i(u,S,T)=\mathbb{E}\left[S_i(\tau(u))\,|\,\tau(u)\leq T\right]=\mathbb{E}\left[\tau(u)\,|\,\tau(u)\leq T\right]\left(r_i-\frac{\sum_{j=1}^d \sigma^2_{i,j}}{\sigma^2} r\right)+\frac{\sum_{j=1}^d \sigma^2_{i,j}}{\sigma^2} u,
\end{equation}
where 
\[\mathbb{E}\left[\tau(u)\,|\,\tau(u)\leq T\right]=\frac{u}{r}\times\frac{\Phi\left(\frac{-u+rT}{\sigma\sqrt{T}}\right)-e^{\frac{2ur}{\sigma^2}}\Phi\left(\frac{-u-rT}{\sigma\sqrt{T}}\right)}{\Phi\left(\frac{-u+rT}{\sigma\sqrt{T}}\right)+e^{\frac{2ur}{\sigma^2}}\Phi\left(\frac{-u-rT}{\sigma\sqrt{T}}\right)}.\]
\end{theorem}
\begin{proof}
First note that the denominator in \(c_i(u,S,T)\) equals \(u\) due to the fact that \(S(\cdot)\) has continuous sample paths. This gives \(K_i(u,S,T)=\mathbb{E}\left[S_i (\tau(u) )\,|\,\tau(u)\leq T\right]\). By introducing the notation \(f_{\tau(u)}(\theta)\) as the probability density function of the first passage time \(\tau(u)\) and \(f_{\tau(u),S_i(\tau(u))}\left(\theta,s\right)\) as the joint density of the first passage time and \(S_i(\tau(u))\), we can write
\begin{align*}
&\mathbb{E}\left[S_i(\tau(u))\,|\,\tau(u)\leq T\right]=\frac{1}{\mathbb{P}\left(\tau(u)\leq T\right)}\int_{0}^T \int_{-\infty}^\infty sf_{\tau(u),S_i(\tau(u))}\left(\theta,s\right){\rm d} s\,{\rm d}\theta.
\end{align*}
The next step is to obtain the density \(f_{\tau(u),S_i(\tau(u))}\left(\theta,s\right)\) from Proposition 5.1 of \cite{Chuang1996} as 
\[f_{\tau(u),S_i(\tau(u))}\left(\theta,s\right)=f_{\tau(u)}\left(\theta\right)f_{S_i(\theta)\,|\,S(\theta)}\left(s\,|\,u\right),\] 
where \(f_{S_i(\theta)\,|\,S(\theta)}\left(s\,|\,u\right)\) denotes the conditional density of \(S_i(\theta)\) given \(S(\theta)\). This is the result of first conditioning on  \(\sigma(S(t),t\leq \theta)\) (i.e.\ the path of \(S(t)_{t\in[0,\theta]}\)) and then making use of the Markov property for correlated Brownian motions as given in Theorem 4.1 of \cite{Chuang1996}. Substituting this into the equation above and using the distributional properties of bivariate Brownian motions we obtain,
\begin{align*}
&\mathbb{E}\left[S_i(\tau(u))\,|\,\tau(u)\leq T\right]=\int_{0}^T \mathbb{E}\left[S_i(\theta)\,|\,S(\theta)=u\right]\frac{f_{\tau(u)}\left(\theta\right)}{\mathbb{P}\left(\tau(u)\leq T\right)}{\rm d}\theta\\
&=\int_{0}^T \left(\left(r_i-\frac{{\mathbb C}{\rm ov}\,(S_i(\theta),S(\theta))}{{\mathbb C}{\rm ov}\,(S(\theta),S(\theta))}r\right)\theta+\frac{{\mathbb C}{\rm ov}\,(S_i(\theta),S(\theta))}{{\mathbb C}{\rm ov}\,(S(\theta),S(\theta))} u\right)\frac{f_{\tau(u)}\left(\theta\right)}{\mathbb{P}\left(\tau(u)\leq T\right)}{\rm d}\theta\\
&=\left(r_i-\frac{\sum_{j=1}^d \sigma^2_{i,j}}{\sigma^2}r\right)\int_{0}^T \frac{\theta f_{\tau(u)}\left(\theta\right)}{\mathbb{P}\left(\tau(u)\leq T\right)}{\rm d}\theta+\frac{\sum_{j=1}^d \sigma^2_{i,j}}{\sigma^2} u\\
&=\left(r_i-\frac{\sum_{j=1}^d \sigma^2_{i,j}}{\sigma^2}r\right)\mathbb{E}\left[\tau(u)\,|\,\tau(u)\leq T\right]+\frac{\sum_{j=1}^d \sigma^2_{i,j}}{\sigma^2} u,
\end{align*}
where we have used that
\begin{equation}\label{eq_BM_tau_cond}
\mathbb{E}\left[\tau(u)\,|\,\tau(u)\leq T\right]= \frac{\int_{0}^T \theta f_{\tau(u)}\left(\theta\right){\rm d}\theta}{\mathbb{P}\left(\tau(u)\leq T\right)}.
\end{equation}
For the denominator in the equation above we note that (see e.g., Theorem 2.1 of \cite{He1998}),
\begin{equation}\label{eq_BM_CDF_tau}
\mathbb{P}\left(\tau(u)\leq T\right)=\mathbb{P}\left(\sup_{0\leq t\leq T } S(t)\geq u\right)=\Phi\left(\frac{-u+rT}{\sigma\sqrt{T}}\right)+e^{\frac{2ur}{\sigma^2}}\Phi\left(\frac{-u-rT}{\sigma\sqrt{T}}\right).
\end{equation}
The density \(f_{\tau(u)}\left(\theta\right)\) can be derived from the above expression by differentiation with respect to \(T\). The numerator in \eqref{eq_BM_tau_cond} is then given by 
\begin{equation}\label{eq_BM_PDF_tau}
\int_{\theta=0}^T \theta f_{\tau(u)}\left(\theta\right){\rm d}\theta=\int_{\theta=0}^T\frac{u}{\sigma\sqrt{2\pi \theta}}e^{-\frac{(u-r\theta)^2}{2\sigma^2\theta}}{\rm d}\theta=\frac{u}{r}\left(\Phi\left(\frac{-u+rT}{\sigma\sqrt{T}}\right)-e^{\frac{2ur}{\sigma^2}}\Phi\left(\frac{-u-rT}{\sigma\sqrt{T}}\right)\right).
 \end{equation}

The final result follows by substituting Equations \eqref{eq_BM_CDF_tau} \& \eqref{eq_BM_PDF_tau} into \eqref{eq_BM_tau_cond}.
\end{proof}

To see that Theorem~\ref{th_BM_alloc} gives the same result as Theorem~\ref{th_wald} in the infinite time horizon regime, we consider the cases \(r<0\) and \(r> 0\) separately. 
\begin{itemize}
\item When \(r<0\), letting \(T\rightarrow\infty\) in \eqref{eq_th_BM_alloc} gives 
\begin{equation}\label{eq_K_BM_inf}
K_i(u,S,\infty)=-\frac{u}{r}\left(r_i-\frac{\sum_{j=1}^d \sigma^2_{i,j}}{\sigma^2}r \right)+\frac{\sum_{j=1}^d \sigma^2_{i,j}}{\sigma^2} u=u\left(\frac{2\sum_{j=1}^d \sigma^2_{i,j}}{\sigma^2} - \frac{r_i}{r}\right).
\end{equation}
For \(r<0\), the supremum of a Brownian motion over an infinite time horizon is exponentially distributed with rate \(-{2r}/{\sigma^2}\) (see Section 6.8 of \cite{Resnick2002}), and thus \[\mathbb{P}\left(\tau(u)<\infty\right)=\mathbb{P}\left(\sup_{t\in[0,\infty)}S(t)>u\right)=e^{\frac{2r}{\sigma^2}u}<1.\]
In this case the existence of \(\vartheta^*\), as the positive solution to \(\ln \mathbb{E}\left[e^{\vartheta S(t)}\right]=(\vartheta r+\frac{1}{2}\vartheta^2\sigma^2)t=0\), is guaranteed and given by \(\vartheta^*=-{2r}/{\sigma^2}\). By taking the derivative, we obtain 
\[m_i=\frac{\partial}{\partial \vartheta_i}e^{\sum_{j=1}^d\vartheta_j r_j+\frac{1}{2}\sum_{j,k=1}^d\vartheta_j\vartheta_k \sigma_{j,k}^2}\bigg|_{{\boldsymbol\vartheta}=\vartheta^*{\boldsymbol 1}}=r_i-\frac{2r}{\sigma^2}\sum_{j=1}^d\sigma_{j,i}^2, \ \ \ \ \ m=\sum_{i=1}^d m_i=-r.\]
This also gives \[K_i(u,S,\infty)=u\left(\frac{2\sum_{j=1}^d \sigma^2_{i,j}}{\sigma^2} - \frac{r_i}{r}\right).\] 
\item In case \(r> 0\), letting \(T\rightarrow\infty\) in \eqref{eq_th_BM_alloc} gives \[K_i(u,S,\infty)=\frac{u}{r}\left(r_i-\frac{\sum_{j=1}^d \sigma^2_{i,j}}{\sigma^2}r \right)+\frac{\sum_{j=1}^d \sigma^2_{i,j}}{\sigma^2} u=u\frac{r_i}{r}.\]
When \(r>0\), we are in the trivial case where \(\mathbb{P}\left(\tau(u)<\infty\right)=1\) and we do not have to do a change of measure (effectively implying that \(\vartheta^*>0\) does not exist). In this case we can apply Wald's identity under the \(\mathbb{P}\)-measure. The existence of the Wald identity for the Brownian case is explicitly mentioned in \cite{Hall1970}. As a result, we find 
\[u=\mathbb{E}\left[S(\tau(u))\,|\,\tau(u)<\infty\right]=\mathbb{E}\left[S(\tau(u))\right]=\mathbb{E}\left[S(1)\right]\mathbb{E}\left[\tau(u)\right]=r\mathbb{E}\left[\tau(u)\right],\] 
which gives 
\(\mathbb{E}\left[\tau(u)\right]=\frac{u}{r}\). By similar reasoning we find
\[K_i(u,S,\infty)=\mathbb{E}\left[S_i(\tau(u))\,|\,\tau(u)<\infty\right]=r_i\mathbb{E}\left[\tau(u)\right]=u\frac{r_i}{r},\]
which coincides with limiting result of \eqref{eq_th_BM_alloc} when \(T\rightarrow\infty\). 
\end{itemize}
Note that in both cases, the allocation fractions over an infinite time horizon \(c_i(u,S,\infty)\) do not depend on the capital level \(u\). This gives rise to a stable, robust allocation method. Note however that this property only holds when considering an infinite time horizon: for finite $T$ there evidently is a dependence on $u$.

\section{Comparison with gradient allocation}\label{sec_GVaR}
In this section we analyze a well-known allocation method, namely the gradient allocation. We start by presenting a number of general results, and then provide explicit results for the Brownian case. 

\subsection{General results}
The gradient allocation method is based on the idea of allocating capital according to the infinitesimal marginal impact of each risk process to the risk measure. In this subsection we will present the gradient allocation for capital determined by the dynamic VaR measure defined in \eqref{eq_VaR}. Under certain conditions, this gradient capital allocation can be expressed as a function of the location of the supremum of the aggregated risk process. This gives rise to a new allocation method. We show that for some risk processes, the \({\rm AVaR}^{\alpha}\) capital allocation method, as defined in the previous section, coincides with the gradient capital allocation for capital determined by the dynamic VaR measure.

\vb 

We proceed with a few words on the existing literature in relation to our work. The gradient allocation approach is also referred to as the Euler allocation method due to its relation with Euler's theorem on homogeneous functions of degree 1. As mentioned, it is based on the idea of allocating risk according to the infinitesimal marginal impact of each individual risk. Several papers have been written on the topic, highlighting the importance and practical use of the gradient allocation method; see e.g., \cite{Tasche1999,Tasche2007}. In \cite{Tasche1999} the author derives an expression for the gradient allocation method applied to the quantile-based risk measure value-at-risk (VaR) for random variables under some smoothness conditions. This VaR risk measure considers the distribution of the sum of random variables, i.e., \[\inf\left\{x\geq 0\,\left|\,\mathbb{P}\left(\sum_{i=1}^dX_i\geq x\right)\leq \alpha\right.\right\},\] and should not be confused with the dynamic VaR presented in Section~\ref{subsec_riskmeasure}. The dynamic VaR considers the probability of ruin over time of an aggregated risk process. The additional time component inherent in the ruin probability complicates matters significantly. We will follow a similar logic as used in \cite{Tasche1999} to find an expression for the gradient allocation applied to the dynamic VaR measure under some smoothness conditions. 

\vb

To properly define the gradient allocation method for the dynamic VaR measure, it is useful to introduce the weight variables \(\boldsymbol{x}=(x_1,\ldots,x_d)^\top\in\mathbb{R}^d\) with \(Z_{i,T}(x_i):=\sup_{t\in[0,T]}\sum_{j\neq i}^dS_j(t)+x_i S_i(t)\) and the function 
\[q_{{\rm VaR}^{\alpha}_{i,T}}(x_i):={\rm VaR}^{\alpha}\left(\sum_{j\neq i}^d S_j+x_iS_i,T\right)={\rm VaR}^{\alpha}\left(Z_{i,T}(x_i),T\right),\]
where \({\rm VaR}^{\alpha}\) is as defined in Equation \eqref{eq_VaR}. When \(q_{{\rm VaR}^{\alpha}_{i,T}}(x_i)\) is differentiable (in \(x_i\)), the gradient allocation method applied to the dynamic VaR type measure is defined by
\[{\rm GVaR}^{\alpha}_i(S,T):=\frac{\partial q_{{\rm VaR}^{\alpha}_{i,T}}(x_i)}{\partial x_i}\Bigg|_{x_i=1}.\]
However, in general the quantile function \(q_{{\rm VaR}^{\alpha}_{i,T}}(x_i)\) will not be differentiable in \(x_i\). By the implicit function theorem, as stated in Appendix \ref{app_A}, the quantile function is differentiable in \(x_i=1\) when the following three conditions hold:
\begin{enumerate}
\item \(\mathbb{P}\left(Z_{i,T}(x_i)\leq q_{{\rm VaR}^{\alpha}_{i,T}}(x_i)\right)=\alpha\) in a neighboorhood of \(x_i=1\).
\item \(\mathbb{P}\left(Z_{i,T}(x_i)\leq y\right)\) is continuously differentiable in \(y\) in an open interval around \(q_{{\rm VaR}^{\alpha}_{i,T}}(x_i)\).
\item \(\mathbb{P}\left(Z_{i,T}(x_i)\leq y\right)\) differentiable in \(x_i\) in an open interval around \(x_i=1\).
\end{enumerate}
Under these conditions the gradient allocation method applied the dynamic VaR measure can be computed as 
\begin{equation}\label{eq_gvar_der}
{\rm GVaR}^{\alpha}_i(S,T)=-\frac{\partial \mathbb{P}\left(Z_{i,T}(x_i)\leq y\right)/\partial x_i}{\partial \mathbb{P}\left(Z_{i,T}(x_i)\leq y\right)/\partial y}\Bigg|_{x_i=1,y=q_{{\rm VaR}^{\alpha}_{i,T}}(x_i)}.
\end{equation}

In Theorem~\ref{th_Euler} we show that under some smoothness assumptions on the multivariate stochastic process \(\boldsymbol{S}(\cdot)=(S_1(\cdot),\ldots,S_d(\cdot))^\top\) , this gradient capital allocation exists and can be expressed in terms of the location of the supremum of aggregated the risk process \(S(\cdot)\).  For this purpose we introduce the set of times at which the supremum is reached:
\[A^*_{i,T}(x_i):=\left\{t\in[0,T]: \sum_{j\neq i}^d S_j(t)+x_iS_i(t)=\sup_{s\in[0,T]}\sum_{j\neq i}^d S_j(s)+x_iS_i(s)\right\},\]
with \(t^*_{i,T}(x_i)\) as the minimum of \(A^*_{i,T}(x_i)\) in case it is non-empty. When \(A^*_{i,T}(x_i)\) is empty we set \(t^*_{i,T}(x_i)=\infty\). Over an infinite horizon \((T=\infty)\) we use similar notation: \(A^*_{i,\infty}(x_i)\) and \(t^*_{i,\infty}(x_i)\), respectively. We often omit the dependence on \(i\) and \(x_i\) when \(x_i=1\), i.e.\ 
\[A^*_{T}:=\left\{t\in[0,T]: \sum_{i=1}^d S_i(t)=\sup_{s\in[0,T]}\sum_{i=1}^d S_j(s)\right\},\]
with \(t^*_{T}\) as the minimum of \(A^*_{T}\). Over an infinite time horizon we adopt the notation \(A^*_{\infty}\) and \(t^*_{\infty}\).
In the case of L\'{e}vy processes excluding compound Poisson processes, the supremum over a finite horizon is obtained at a unique point in time almost surely, i.e., when \(A_{i,T}^*(x_i,\omega)\) is non-empty it is a singleton a.s.\ (see page 171 in \cite{Kyprianou2006}).  Before proceeding with the stochastic case we first introduce a lemma concerning the differentiability of functions.  In view of the applications in Section~\ref{sec_examples}, we content ourselves with functions and processes that have continuous paths (as in the Brownian case) and those that have upward jumps with drift (as in the Compound Poisson case with drift).

\begin{lemma}\label{lem_der_fun}
Let \(f:[0,T]\times \mathbb{R}_{>0}\mapsto \mathbb{R}\) with \(f(t,x):=p(t)+xq(t)\) and assume that for all positive \(x\) the maximum of \(t\mapsto f(t,x)\) over \(0\leq t\leq T\) equals the supremum and is uniquely attained in \(t_x\) with value \(\bar f(x):=f(t_x,x)\).  If one of the following two assumptions is satisfied:
\begin{enumerate}
\item \(p(t)\) and \(q(t)\) are continuous functions of \(t\),
\item \(p(t)\) and \(q(t)\) are functions for which the location of the maximum lies within a finite set of points for all positive \(x\),  i.e.\ for all positive \(x\) we have \(t_x\in\mathcal{T}\) for a finite set \(\mathcal{T}:=\{T_1,\ldots,T_n\}\). The time points \(T_i\in[0,T]\) are independent of \(x\). 
\end{enumerate}
Then 
\[\frac{\partial}{\partial x}\bar f(x)=q(t_x).\]
\end{lemma}
\begin{proof}
We will treat the two assumptions separately, starting with the first assumption. Under assumption 1, the functions \(f(t,x)\) and \(\partial f(t,x)/\partial x=q(t)\) are continuous in both \(t\) and \(x\).  We may now apply Danskin's Min-Max Theorem (see Theorem 1 in \cite{Danskin1966}).  According to this theorem, the right derivative of \(\bar f(x)\) in \(x\) is given by \(q(t_x)\), where we have used that the location of the maximum is uniquely obtained in \(t_x\). The left derivative is minus the directional derivative in the direction of the \(-x\) axis, which also gives \(q(t_x)\) by the unique location of the maximum. The final result now follows.

We now proceed with a proof of the result under assumption 2. Following the proof of Proposition 2.1 in \cite{Oyama2018} we have
\[
f(t_x,x+z)-f(t_x,x)\leq \bar f(x+z)- \bar f(x)\leq f(t_{x+z},x+z)-f(t_{x+z},x),
\]
which can be rewritten as 
\begin{align*}
\lefteqn{p(t_x)+(x+z)q(t_x)-\left(p(t_x)+xq(t_x)\right)  \leq \bar f(x+z)- \bar f(x)} \\
& \qquad \leq p(t_{x+z})+(x+z)q(t_{x+z})-\left(p(t_{x+z})+xq(t_{x+z})\right),
\end{align*}
i.e.,
\[
zq(t_x)\leq \bar f(x+z)- \bar f(x)\leq zq(t_{x+z}).
\]
This gives,  
\[\begin{cases}
q(t_x)\leq \frac{\bar f(x+z)- \bar f(x)}{z}\leq q(t_{x+z}), & {\rm for \  } z>0\\
q(t_{x+z})\leq \frac{\bar f(x+z)- \bar f(x)}{z}\leq q(t_{x}), & {\rm for \  } z<0.
\end{cases}\]
The final result follows when \(q(t_{x+z})\to q(t_x)\) for \(z\to 0\). We will now show this holds true under assumption 2. Under this assumption, the maximum of \(f(t,x)\) is obtained in one of the finite number of points \(\mathcal{T}:=\{T_1,\ldots,T_n\}\), i.e.\ \(t_x\in\mathcal{T}\). As the maximum is considered to be uniquely obtained, \(\bar{f}(x)-f(t,x)>0\) for all \(t\in \mathcal{T}\setminus \{t_x\}\). We now consider the function \(f(t,x+z)\) and its unique maximum \(t_{x+z}\in\mathcal{T}\). We set
\[\delta=\min_{t\in T_q}\bigg|\frac{\bar{f}(x)-f(t,x)}{q(t_{x})-q(t)}\bigg|, \ \ \ T_q=\{t\in\mathcal{T}\setminus \{t_x\}: q(t)\neq q(t_x)\}.\]
In case \(T_q\) is empty, we set \(\delta=1\). For small \(|z|\leq \delta\)
 and all \(t\in \mathcal{T}\setminus \{t_x\}\), we have \(\bar{f}(x)>f(t,x)\) and 
\[f(t_{x},x+z)-f(t,x+z)=\bar{f}(x)-f(t,x)+z(q(t_{x})-q(t))>0.\]
In other words, for small \(|z|\), the maximum of \(f(t,x+z)\) is obtained in \(t_{x+z}=t_x\) and thus \(q(t_{x+z})= q(t_x)\). 
\end{proof}

L\'{e}vy processes \(S_i(t)\) with sample paths of the types \(q(t\) (or similarly \(p(t)\)) as specified in Lemma~\ref{lem_der_fun} include continuous L\'{e}vy processes, i.e. scaled Brownian motions with drift, and compound Poisson processes with non-zero drift. These processes are used in the examples in Section~\ref{sec_examples}. The proof of the next theorem can be found in Appendix \ref{app_B}.

\begin{theorem}\label{th_Euler}
Consider multivariate L\'{e}vy processes \(\boldsymbol{S}(\cdot)=(S_1(\cdot),\ldots,S_d(\cdot))\) on the probability space \((\Omega,\mathcal{F},\mathbb{P})\). For small \(\delta,\bar{\delta}>0\) and \(x_i\in [1-\delta,1+\delta]\), assume that the individual processes \(S_i(\cdot)\) and the aggregated process \(\sum_{j\neq i}^dS_j(\cdot)+x_i S_i(\cdot)\) are either continuous, i.e. scaled Brownian motions with drift, or compound Poisson processes with negative drift and positive jumps. Denote \(y_i \rightarrow f_{i,x_i}(y_i)\) as the density of the random variable \(Z_{i,T}(x_i):=\sup_{t\in[0,T]}\sum_{j\neq i}^dS_j(t)+x_i S_i(t)\) at the point \(y_i\in[q_{{\rm VaR}^{\alpha}_{i,T}}(x_i)-\bar{\delta},q_{{\rm VaR}^{\alpha}_{i,T}}(x_i)+\bar{\delta}]\) which we assume exists. Furthermore, if the following conditions hold,
\begin{enumerate}[(i)]
\item The density \(f_{i,x_i}(y_i)<\infty\) is continuous in \(y_i,x_i\) for all \(x_i\in [1-\delta,1+\delta]\) and \(y_i\in[q_{{\rm VaR}^{\alpha}_{i,T}}(x_i)-\bar{\delta},q_{{\rm VaR}^{\alpha}_{i,T}}(x_i)+\bar{\delta}]\).
\item For each \(x_i\), \(f_{i,x_i}(q_{{\rm VaR}^{\alpha}_{i,T}}(x_i))>0\).
\item \(\mathbb{E}\left[S_i(t_{i,T}^*(x_i))\,\big|\,Z_{i,T}(x_i)=y_i\right]\) and \(\mathbb{E}\left[|S_i(t_{i,T}^*(x_i))|\,\big|\,Z_{i,T}(x_i)=y_i\right]\) exist and are continuous in \(y_i,x_i\) for all \(x_i\in [1-\delta,1+\delta]\) and \(y_i\in[q_{{\rm VaR}^{\alpha}_{i,T}}(x_i)-\bar{\delta},q_{{\rm VaR}^{\alpha}_{i,T}}(x_i)+\bar{\delta}]\).
\item \(\mathbb{E}\left[\sup_{t\in[0,T]}\big|S_i(t)\big|\right]<\infty\).
\end{enumerate}
then the gradient allocation method applied to the dynamic VaR type risk measure, i.e.\ gradient capital allocation, exists and is given by
\begin{equation}\label{eq_GVaR}
{\rm GVaR}^{\alpha}_i(S,T)=\mathbb{E}\left[S_i(t_T^*)\,\bigg|\,\sup_{t\in[0,T]}S(t)={\rm VaR}^{\alpha}(S,T)\right].
\end{equation}
\end{theorem}
Especially when one cannot analytically determine the gradient capital allocation by differentiation as in \eqref{eq_gvar_der}, and one has to resort to simulation, the above result provides a practical way to determine the gradient capital allocation. Numerical evaluation of the gradient capital allocation \eqref{eq_gvar_der} can be computationally expensive due to the derivatives. For the Gaussian case, without the additional time component in the determination of capital, a similar result has been derived (see \cite{Tasche1999}, and \cite{Gourieroux2000}) and has found broad application in practice.  In fact, expression \eqref{eq_GVaR} can be seen as the extension of the results found in \cite{Tasche1999} (Lemma 5.3 and Remark 5.4) and \cite{Gourieroux2000} with respect to time. Note that there is a similarity between assumptions (i)-(iii) in Theorem~\ref{th_Euler} and the assumptions imposed in \cite{Tasche1999}. 

\begin{remark}
We note that the scope of Theorem~\ref{th_Euler} could be extended to multivariate stochastic processes in a natural way whenever the sample paths of the processes are of the types specified in Lemma~\ref{lem_der_fun}. This includes for example Markov-modulated Brownian motions.  Furthermore, we note that the negative drift assumption for compound Poisson processes can be replaced by positive drift as well by noting that for positive drifts the supremum can be reached in the time points of the jumps of the Poison process or at the final time horizon \(T\).
\end{remark}

Expression \eqref{eq_GVaR} suggests yet another allocation method for the allocation of the capital reserve level \(u\):
\begin{equation}\label{eq_K_tilde_allocation}
\overline{K}_i(u,S,T):=\mathbb{E}[S_i(t_T^*)\,|\,S(t_{T}^*)=u,t_{T}^*\leq T],\:\:\:\:\overline{c}_i(u,S,T):=\frac{\overline{K}_i(u,S,T)}{u}.
\end{equation}
As before, for the infinite time horizon we define \[\overline{K}_i(u,S,\infty):=\mathbb{E}[S_i(t_{\infty}^*)\,|\,S(t_\infty^*)=u,t_{\infty}^*<\infty],\:\:\:\:\overline{c}(u,S,\infty):=\frac{\overline{K}_i(u,S,\infty)}{u}.\]
Note that this allocation method cannot be evaluated when \(t_{T}^*=\infty\) almost surely, i.e., when \(A^*_{T}\) is empty. When it can be evaluated, the allocation method also satisfies the properties highlighted in Section~\ref{subsec_AVaR}, i.e., the allocated risk measure is positively homogeneous, deterministic when the marginal risk is deterministic and admits the full allocation principle. 

\vb

When the capital level \(u\) is determined using the dynamic VaR measure, this allocation method gives for component \(i\),
\[\avarbar:=\overline{K}_i\left({\rm VaR}^{\alpha}(S,T),S,T\right).\]
Whenever the conditions of Theorem~\ref{th_Euler} are satisfied, we find \({\rm GVaR}^{\alpha}_i(S,T)=\avarbar\). 

\subsection{Special case: Brownian motion}\label{subsec_BM_GVAR}
In this subsection we show the implications of the results presented in the previous subsection for the special case of scaled Brownian motion with drift. Throughout we use the same notation as introduced in Section~\ref{subsec_BM_AVAR}. We start by deriving an explicit expression for the new allocation method \eqref{eq_K_tilde_allocation} in the theorem below. 

\begin{theorem}\label{th_BM_alloc_max}
For the multivariate scaled Brownian motion with drift and finite time horizon \(T\), the allocated capital reserve of component \(i \in \{1,\ldots,d\}\) is given by
\[\overline{K}_i(u,S,T)=\mathbb{E}\left[S_i\left(t^*_{T}\right)|\,S(t^*_{T})=u\right]=\mathbb{E}\left[t^*_T\,|\,S(t^*_{T})=u\right]\left(r_i-\frac{\sum_{j=1}^d \sigma^2_{i,j}}{\sigma^2} r\right)+\frac{\sum_{j=1}^d \sigma^2_{i,j}}{\sigma^2} u,\]
where 
\[\mathbb{E}\left[t^*_T\,|\,S(t^*_T)=u\right]=u\left(-r+\frac{\sigma e^{-\frac{(u+rT)^2}{2\sigma^2 T}}}{\sqrt{2\pi T}\Phi\left(\frac{-u-r T}{\sigma\sqrt{T}}\right)}\right)^{-1}.\]
\end{theorem}
\begin{proof}

For Brownian processes considered over a finite time horizon \(T\), \(A^*_T\) is non-empty and \(t^*_T\) is thus properly defined due to the fact that for Brownian motions the supremum is actually attained (i.e., it is equal to the maximum). In fact, \(t_T^*\) is finite and almost surely unique for finite \(T\) (see page 158 in \cite{Kyprianou2006}, or Lemma 49.4 in \cite{Sato1999}). As a result, we find 
\[\overline{K}_i(u,S,T)=\mathbb{E}\left[S_i\left(t^*_T\right)|\,S(t^*_T)=u\right].\]
We introduce the notation \(f_{S(t_T^*),t_T^*}(x,\theta)\), \(f_{t_T^*\,|\,S(t^*_T)}\left(\theta\,|\,x\right)\), and \(f_{S(t_T^*)}(x)\) to denote the joint probability density function of \(S(t_T^*)\) and \(t_T^*\), the conditional density function of \(t_T^*\) conditional on \(S(t^*_T)\), and the density function of the maximum of the process \(S(\cdot)\), respectively, which are properly defined for \(0\leq x<\infty\) and \(0\leq \theta\leq T\).
By conditioning on the location of the maximum \(t_T^*\) and using standard results for the conditional mean of a bivariate Brownian motion, we obtain
\begin{align*}
&\mathbb{E}\left[S_i(t^*_T)\,|\,S(t^*_T)=u\right]=\int_{\theta =0}^T \mathbb{E}\left[S_i(\theta)\,|\,S(\theta)=u\right]f_{t_T^*\,|\,S(t^*_T)}\left(\theta\,|\,u\right){\rm d}\theta\nonumber\\
&=\int_{\theta=0}^T \left(\left(r_i-\frac{{\mathbb C}{\rm ov}\,(S_i(\theta),S(\theta))}{{\mathbb C}{\rm ov}\,(S(\theta),S(\theta))}r\right)\theta+\frac{{\mathbb C}{\rm ov}\,(S_i(\theta),S(\theta))}{{\mathbb C}{\rm ov}\,(S(\theta),S(\theta))} u\right)f_{t_T^*|S(t^*_T)}\left(\theta \,|\,u\right){\rm d}\theta\nonumber\\
&=\left(r_i-\frac{\sum_{j=1}^d \sigma^2_{i,j}}{\sigma^2}r\right)\int_{\theta=0}^T \theta f_{t_T^*\,|\,S(t^*_T)}\left(\theta\,|\,u\right){\rm d}\theta+\frac{\sum_{j=1}^d \sigma^2_{i,j}}{\sigma^2} u\int_{\theta=0}^T f_{t_T^*\,|\,S(t^*_T)}\left(\theta\,|\,u\right){\rm d}\theta\nonumber\\
&=\left(r_i-\frac{\sum_{j=1}^d \sigma^2_{i,j}}{\sigma^2}r\right)\mathbb{E}\left[t^*_T\,|\,S(t_T^*)=u\right]+\frac{\sum_{j=1}^d \sigma^2_{i,j}}{\sigma^2} u.
\end{align*}
For the conditional expectation we have 
\begin{equation}\label{eq_E_t_star}
\mathbb{E}\left[t^*_T\,|\,S(t_T^*)=u\right]=\frac{\int_0^T \theta f_{S(t_T^*),t_T^*}(u,\theta){\rm d}\theta}{f_{S(t_T^*)}(u)}.
\end{equation}
The density function \(f_{S(t_T^*)}(u)\) can be derived from its cumulative distribution function as presented in e.g., \cite{He1998}, i.e.,
\begin{equation}\label{eq_density_maxBM}
f_{S(t_T^*)}(u)=\frac{1}{\sigma\sqrt{T}}\phi\left(\frac{u-r T}{\sigma\sqrt{T}}\right)-\frac{2r}{\sigma^2}e^{\frac{2ur}{\sigma^2}}\Phi\left(\frac{-u-r T}{\sigma\sqrt{T}}\right)+\frac{1}{\sigma\sqrt{T}}e^{\frac{2ur}{\sigma^2}}\phi\left(\frac{-u-r T}{\sigma\sqrt{T}}\right),
\end{equation}
where \(\Phi(\cdot)\) and \(\phi(\cdot)\) denote the cumulative distribution and probability density function of a standard normal random variable, respectively.
We now focus on finding an expression for \(\int_0^T \theta f_{t_T^*,S(t_T^*)}(\theta,u)\,{\rm d}\theta\). The joint density \(f_{t_T^*,S(t_T^*)}(\theta,u)\) is known: as given in \cite{Shepp1979}, \[f_{t_T^*,S(t_T^*)}(\theta,u)=\int_{-\infty}^u\frac{1}{\pi \sigma^4}\frac{u(u-x)}{\theta^{3/2}(T-\theta)^{3/2}}e^{\left(-\frac{u^2}{2\theta}-\frac{(u-x)^2}{2(T-\theta)}+rx-\frac{r^2T}{2}\right)\sigma^{-2}}dx.\]
This gives
\begin{align}
\int_0^T \theta f_{t_T^*,S(t_T^*)}(\theta,u){\rm d}\theta&=\int_{-\infty}^u\int_0^T\frac{1}{\pi \sigma^4}\frac{u(u-x)}{\theta^{1/2}(T-\theta)^{3/2}}e^{\left(-\frac{u^2}{2\theta}-\frac{(u-x)^2}{2(T-\theta)}+rx-\frac{r^2T}{2}\right)\sigma^{-2}}{\rm d}\theta dx\nonumber\\
&=\int_{-\infty}^u\frac{\sqrt{2}u}{\sigma^3\sqrt{\pi T}}e^{\left(-\frac{(2u-x)^2}{2T}+rx-\frac{r^2T}{2}\right)\sigma^{-2}} dx=\frac{2u}{\sigma^2}e^{\frac{2ur}{\sigma^2}}\int_{-\infty}^u\frac{1}{\sigma\sqrt{2\pi T}}e^{\frac{-(x-2u-rT)^2}{2\sigma^2 T}} dx\nonumber\\
&=\frac{2u}{\sigma^2}e^{\frac{2ur}{\sigma^2}}\Phi\left(\frac{-u-rT}{\sigma\sqrt{T}}\right),\label{eq_exp_MAXBM}
\end{align}
where we have made use of the identity, for \(a>0, \ b>0\),
\[\int_0^T \frac{ab}{\pi\theta^{1/2}(T-\theta)^{3/2}}e^{-\frac{a^2}{2\theta}-\frac{b^2}{2(T-\theta)}}{\rm d}\theta=\frac{2a}{\sqrt{2\pi T}}e^{-\frac{(a+b)^2}{2T}}.\]
This identity is easy to check by taking the Laplace transforms (with respect to \(T\), that is) of both sides using (5.28) and (5.30) on page 41 of \cite{Oberhettinger1973}. Substituting \eqref{eq_density_maxBM} and \eqref{eq_exp_MAXBM} into \eqref{eq_E_t_star} gives the final result.
\end{proof}

The above theorem only considers a finite time horizon as this guarantees that the maximum of the process is almost surely attained. For an infinite time horizon, we separately consider the two instances \(r\geq 0\) and \(r<0\). In the first instance we have \(\lim_{t\rightarrow\infty}S(t)=\infty\) and the supremum of the Brownian process \(S(\cdot)\) is infinite and thus not attained at a finite point in time. In other words, \(t_\infty^*=\infty\). In the second instance, as pointed out in Section~\ref{subsec_BM_AVAR}, the maximum of the process \(S(\cdot)\) is an exponentially distributed random variable. Taking the 
limit of Theorem~\ref{th_BM_alloc_max} with respect to \(T\rightarrow\infty\) for \(r<0\), we find \(\mathbb{E}\left[t^*_\infty\,|\,S(t^*_\infty)=u,t_\infty^*<\infty\right]=-{u}/{r}\) and subsequently 
\[\mathbb{E}\left[S_i\left(t^*_\infty\right)\,|\,S(t^*_\infty)=u,t_\infty^*<\infty\right]=u\left(\frac{2\sum_{j=1}^d \sigma^2_{i,j}}{\sigma^2} -\frac{r_i}{r}\right).\]
Note that this result coincides with the allocation method presented in Section~\ref{sec_Aallocation}, i.e., Equation \eqref{eq_K_BM_inf}. This connection is a consequence of the fact that \(\mathbb{E}\left[t_\infty^*\,|\,S(t_\infty^*)=u,t_\infty^*<\infty\right]=\mathbb{E}\left[\tau(u)\,|\,\tau(u)<\infty\right]\). As we are in the continuous case, note that \(S(\tau(u))=u\). Furthermore, conditional on \(S(t_T^*)=u\), the location of the maximum, \(t_T^*\), is the same as the first passage time of the level \(u\), \(\tau(u)\) conditioned on \(\sup_{\tau(u)\leq t\leq T}S(t)= u\). At the first time the supremum is attained we assume \(S(t_T^*)=u\) and thus \(S(\tau(u))\ngtr u\). As a result we find
\[\mathbb{E}\left[t_T^*\,|\,S(t_T^*)=u,t_T^*\leq T\right]=\mathbb{E}\left[\tau(u)\,\left|\,\tau(u)\leq T,\sup_{\tau(u)\leq t\leq T}S(t)= u\right.\right].\] For finite \(T\), the random variable \(\sup_{\tau(u)\leq t\leq T}S(t)\) is dependent on \(\tau(u)\) as it impacts the length of the time horizon over which the supremum is considered. Considering an infinite time horizon \((T=\infty)\), the equality becomes 
\begin{equation}\label{eq_t^*_tau}
\mathbb{E}\left[t_\infty^*\,|\,S(t_\infty^*)=u,t_T^*<\infty\right]=\mathbb{E}\left[\tau(u)\,\left|\,\tau(u)<\infty,\sup_{\tau(u)\leq t<\infty}S(t)=u\right.\right].
\end{equation} 
The random variable \(\sup_{\tau(u)\leq t<\infty}S(t)\) is independent of \(\tau(u)\) and is therefore redundant in the latter conditional expectation.

\vb

Theorem~\ref{th_Euler} points out the relationship between the gradient capital allocation method and the allocation method \(\overline{K}(u,S,T)\) as introduced in \eqref{eq_K_tilde_allocation}. We will now show, by easy computation, that for scaled Brownian motions with drift, the conditions (i)-(iv) of Theorem~\ref{th_Euler} are satisfied. As a result, the gradient capital allocation is given by \eqref{eq_GVaR} and Theorem~\ref{th_BM_alloc_max} can be used to make this explicit.  We use the same numbering as in Theorem~\ref{th_Euler}. 

\begin{enumerate}[(i)]
\item First note that, for fixed \(T\), the drift and variance of the random variable \[Z_{i,T}(x_i):=\sup_{t\in[0,T]}\sum_{j\neq i}^dS_j(t)+x_i S_i(t)\] 
are given by \(r(x_i)=r+(x_i-1)r_i\) and \(\sigma^2(x_i)=\sigma^2+2(x_i-1)\sum_{j=1}^d\sigma^2_{j,i}+(x_i-1)^2\sigma^2_{i,i}\), respectively. Both are continuous functions in \(x_i\). The density of the supremum process of a scaled Brownian motion with drift \(r\) and variance parameter \(\sigma^2\) was given in \eqref{eq_density_maxBM}, which can be seen to be a continuous function of \(r\) and \(\sigma^2\). As a result, the density \(f_{i,x_i}(\cdot)\) is also continuous in \(x_i\). Furthermore, the density is finite for finite \(u, r(x_i)\) and finite \(\sigma(x_i),T>0\).
\item The density, as mentioned in the item above, is also greater than zero for finite \(u, r(x_i)\) and finite \(\sigma(x_i),T>0\).
\item The conditional expectation \(\mathbb{E}\left[S_i(t_{i,T}^*(x_i))\,|\,Z_{i,T}(x_i)=y_i\right]\) can be made explicit by use of Theorem~\ref{th_BM_alloc_max} and can be seen to be continuous in \(x_i\) and \(y_i\). Using a similar approach as in Theorem~\ref{th_BM_alloc_max}, the conditional expectation \(\mathbb{E}\left[|S_i(t_{i,T}^*(x_i))|\,\big|\,Z_{i,T}(x_i)=y_i\right]\) can be derived explicitly by conditioning on the location of the maximum and using Lemma~\ref{lem_bi_BM}. The resulting expression is continuous in both \(x_i\) and \(y_i\) due to the continuity of all the functions involved.
\item Note that \[\mathbb{E}\left[\sup_{t\in[0,T]}\big|S_i(t)\big|\right]\leq |r_i|T+\sigma_i\mathbb{E}\left[\sup_{t\in[0,T]}\big|B(t)\big|\right],\] where \(B(t)\) is a standard Brownian motion. For a standard Brownian motion it can be derived that \[\mathbb{E}\left[\sup_{t\in[0,T]}\big|B(t)\big|\right]=\sqrt{\frac{\pi}{2}};\] this follows by integrating the tail probabilities of $\sup_{t\in[0,T]}\big|B(t)\big|$ which can be found in \cite{Borodin2002} (Part II, Chapter 3, Formula 1.1.4). As a result, \(\sup_{t\in[0,T]}\big|S_i(t)\big|\) has a finite mean when \(|r_i|,\sigma_i<\infty\).
\end{enumerate}

It should be pointed out that in the Brownian case, the gradient allocation method applied to the dynamic VaR measure can also be derived explicitly by straightforward differentiation of the dynamic VaR measure. Considering an infinite horizon \((T=\infty)\) and \(r<0\), the supremum is exponentially distributed with parameter \({-2r}/{\sigma^2}\). More specifically, we find \({\rm VaR}^{\alpha}(S,\infty)=\frac{\sigma^2}{2r}\ln(\alpha)\). Furthermore, by differentiation the gradient capital allocation can be determined as 
\[{\rm GVaR}^{\alpha}_i(S,\infty)=\ln(\alpha)\frac{2r^2\sum_{j=1}^d\sigma_{j,i}^2-r_ir\sigma^2}{2r^3}, \ \ \ \ \ \ \ \frac{{\rm GVaR}^{\alpha}_i(S,\infty)}{{\rm VaR}^{\alpha}(S,\infty)}=\frac{2\sum_{j=1}^d\sigma_{j,i}^2}{\sigma^2}-\frac{r_i}{r},\]
which is in line with the results obtained before.

\vb

We conclude this section by briefly mentioning an appealing property of the gradient allocation method for Brownian risk processes. Consider two Brownian risk processes \(S_1(\cdot)\) and \(S_2(\cdot)\), we find, by an elementary computation, \({\rm VaR}^{\alpha}(S_1,\infty)+{\rm VaR}^{\alpha}(S_2,\infty)\geq {\rm VaR}^{\alpha}(S_1+S_2,\infty)\). This is known as the sub-additivity property, see Axiom 4 in Section~\ref{subsec_riskmeasure}. By Theorem 3.1 of \cite{Buch2008} it then follows that the so-called `no undercut' property is satisfied for the gradient allocation of this risk measure, i.e., \[\sum_{i\in N}{\rm GVaR}^{\alpha}_i(S,\infty)\leq {\rm VaR}^{\alpha}\left(\sum_{i\in N}S_i,\infty\right)\] for all subsets \(N\) of \(\{1,\ldots,d\}\).

\section{Numerical examples}
\label{sec_examples}
In this section we present a series of illustrative examples featuring the allocation methods presented in the previous sections. 

\subsection{Brownian motion}\label{ex_BM}
This example considers the multivariate Brownian motion process as in Sections \ref{subsec_BM_AVAR} and \ref{subsec_BM_GVAR}, adopting the same notation. In those sections we have derived explicit expressions for the two proposed capital allocation methods, \(K_i(u,S,T)\) and \(\overline{K}_i(u,S,T)\). In particular, we considered the instance where capital \(u\) is determined by the dynamic VaR measure and its corresponding allocations \({\rm AVaR}^{\alpha}_i(S,T)\) and \(\avarbar\). The \(\avarbar\) allocation has been shown to coincide with the gradient allocation method \({\rm GVaR}^{\alpha}_i(S,T)\) in Section~\ref{subsec_BM_GVAR}. The same section also illustrates that the two newly proposed allocation methods coincide when considering an infinite time horizon,  i.e., \({\rm AVaR}^{\alpha}_i(S,\infty)= \overline{\rm AVaR}^{\alpha}_i(S,\infty)\). These capital allocation fractions, when considering an infinite time horizon, do not depend on the capital level \(u\) (or \(\alpha\) when using the dynamic VaR measure to determine capital).
In this subsection we assess the difference between the two new allocation methods over a finite time horizon and their sensitivities towards the capital reserve level \(u\).

\vb

We consider a simplified setting with two risk processes \(S_1(\cdot)\) and \(S_2(\cdot)\). Both processes have unit variance \(\sigma_{1,1}=\sigma_{2,2}=1\) and the correlation between the Brownian motions is given by \(\rho=0.5\). We set \(r=(-2,-1)\) as the drift vector, ensuring that the negative drift assumption is satisfied. For the infinite time horizon we obtain the allocation fractions \({m_1}/{m}=c_1(u,S,\infty)=\overline{c}_1(u,S,\infty)=\frac{1}{3}\) and \(m_2/m=\frac{2}{3}\). In Figure \ref{fig_BM_allocation} we have plotted the allocation fraction of the first risk process \(S_1(\cdot)\) for the newly proposed allocation methods, i.e.\ \({\rm AVaR}^\alpha_1(S,T)/{\rm VaR}^\alpha_1(S,T)\) and \(\avarbare/{\rm VaR}^\alpha(S,T)\) as a function of the time horizon \(T\) for various values of \(\alpha\). Lower values of \(\alpha\) correspond to higher values of \({\rm VaR}^\alpha(S,T)\) and vice versa. The figure illustrates that the allocation methods align for very short time horizons (for which the risk is divided equally) and converge to the same limit, \(\frac{1}{3}\) for the first risk process (or business line), over a long time horizon. For the intermediate time horizons, the two allocation methods differ and the difference becomes more substantial when \(\alpha\) increases. 

\vb

In Figure \ref{fig_BM_allocation_2} we have again plotted, for the first risk process (or business line), the two newly proposed allocation fractions but now allocating general capital reserve level \(u\) instead of \({\rm VaR}^\alpha(S,T)\). The left panel shows \(c_1(u,S,T)\) and \(\overline{c}_1(u,S,T)\) as a function of \(u\) for various time horizons \(T\). The allocations become less sensitive to the capital level \(u\) as the time horizon \(T\) increases. The right panel of Figure \ref{fig_BM_allocation_2} illustrates the allocation fraction as a function of the time horizon for positive drift \(r=(2,1)\). In this case, the allocation fraction \(c_1(u,S,T)\) converges towards \(\frac{r_1}{r}=\frac{2}{3}\) when \(T\rightarrow\infty\), as has been pointed out in Section~\ref{subsec_BM_AVAR}. The allocation \(\overline{c}_1(u,S,T)\), however, does not. As discussed in Section~\ref{subsec_BM_GVAR}, this allocation method is not properly defined for an infinite time horizon in case of positive drifts. When \(u\) is large relative to the time horizon (and the other parameters), the two allocation methods are comparable. This is a result of the fact that for relatively large \(u\), both the expectations \(\mathbb{E}\left[t_T^*\,|\,S(t_T^*)=u\right]\) and \(\mathbb{E}\left[\tau(u)\,|\,\tau(u)\leq T\right]\) approach \(T\) in case of positive drifts. 

\begin{figure}[H]
\centering
\includegraphics[scale=0.75]{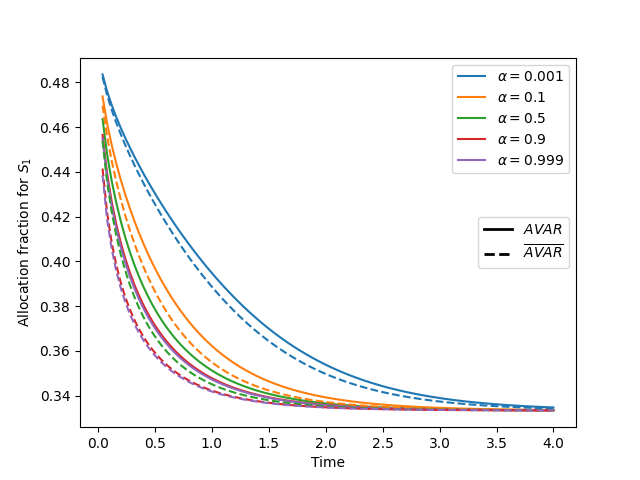}
\caption{Allocation fractions of the first risk process. The figure shows (solid lines) \({\rm AVaR}^\alpha_1(S,T)/{\rm VaR}^\alpha(S,T)=c_1({\rm VaR}^\alpha(S,T),S,T)\) and (dashed lines) \({\rm GVaR}^\alpha_1(S,T)/{\rm VaR}^\alpha(S,T)=\overline{c}_1({\rm VaR}^\alpha(S,T),S,T)\) as a function of the time horizon \(T\) for multiple values of \(\alpha\).}
\label{fig_BM_allocation}
\end{figure}

\begin{figure}[H]
\resizebox{8cm}{6cm}{
\includegraphics[scale=0.75]{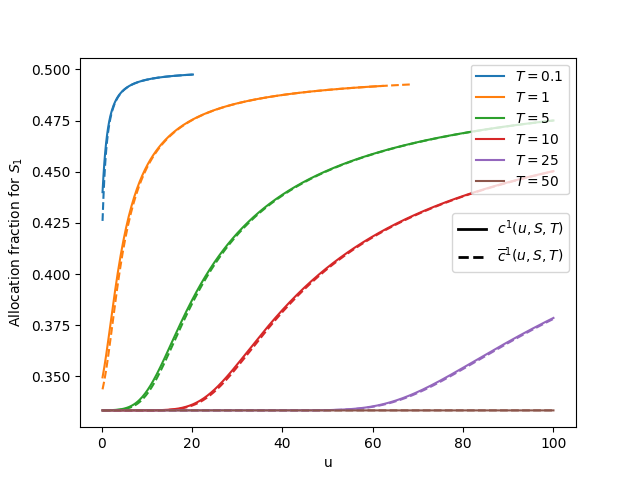}}
\resizebox{8cm}{6cm}{
\includegraphics[scale=0.75]{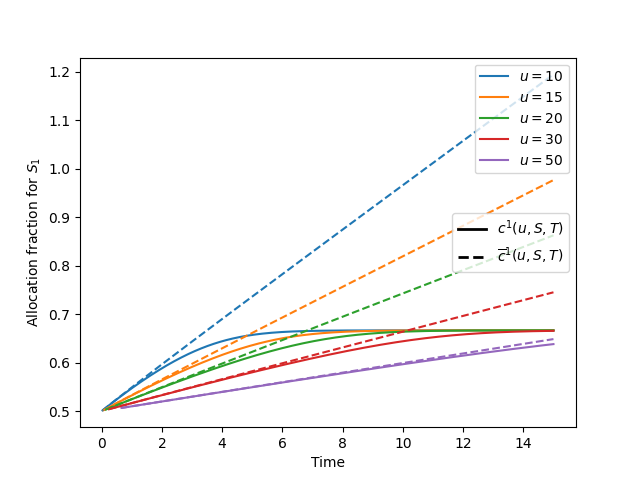}}
\caption{Left panel: Allocation fractions of the first risk process, \(c_1(u,S,T)\) and \(\overline{c}_1(u,S,T)\), as functions of \(u\) for multiple values of \(T\). Right panel: Allocation fractions of the first risk process \(c_1(u,S,T)\) and \(\overline{c}_1(u,S,T)\), as functions of \(T\) for multiple values of \(u\) in case of positive drift parameters \(r=(2,1)\).}
\label{fig_BM_allocation_2}
\end{figure}

\subsection{Spectrally negative L\'{e}vy process}\label{ex_spec_neg}
Consider a spectrally negative multivariate L\'{e}vy process \(\boldsymbol{S}(\cdot)=(S_1(\cdot),\ldots,S_d(\cdot))^\top \). This means that the process does not contain positive jumps, i.e., the L\'{e}vy measure has support on \((-\infty,0]\) only. For more general properties of spectrally negative L\'{e}vy processes we refer the reader to Chapter VII in \cite{Bertoin1996} or Chapter 8 in \cite{Kyprianou2006}. The aggregated process \(S(\cdot)\) is also a spectrally negative L\'{e}vy process. Equivalently, the paths are skip-free upwards and ruin can only be caused by the drift and diffusion parts. Due to the absence of upwards jumps, we have \(S(\tau(u))=u\) almost surely (given that \(\tau(u)<\infty\)). In this example the focus lies on the first proposed allocation method \(K_i(u,S,T)\) and its quantification according to Theorem~\ref{th_wald} and Remark~\ref{rem_Hall}. 

\vb

In the spectrally negative case, it is known that the exponential moments of the L\'{e}vy process are finite for all real \(\vartheta\geq 0\): for \(\vartheta\geq 0\) we have \(\mathbb{E}[e^{\vartheta S(t)}]<\infty\). As a result, for \(\vartheta\geq 0\), the function \(\kappa(\vartheta)\) is the Laplace exponent which is strictly convex by Holder's inequality and \(\lim_{\vartheta\rightarrow\infty}\kappa(\vartheta)=\infty\). As before, we will assume a negative drift \(\mathbb{E}[S(1)]=\kappa' (0)<0\) (i.e., default can only be caused by the diffusion part), so as to rule out the situation of almost sure ruin. It then follows that \(\kappa(\vartheta)\) has a real positive zero \(\vartheta^*>0\). See also Theorem XI.2.3.\ in \cite{Asmussen2010}.
 
By similar reasoning there exists a real \(\vartheta>0\) such that \(\kappa^{\mathbb{Q}}(-\vartheta)=\kappa(-\vartheta+\vartheta^*)\leq 0\) and thus \(\mathbb{E}^{\mathbb{Q}}\left[e^{-\theta S(1)}\right]=e^{\kappa(-\theta+\vartheta^*)}\leq 1\). By Remark~\ref{rem_Hall}, \(\mathbb{E}^{\mathbb{Q}}\left[S(\tau(u))\right]\) exists and is given by
\[\mathbb{E}^{\mathbb{Q}}\left[S(\tau(u))\right]=\mathbb{E}^{\mathbb{Q}}\left[S(1)\right]\mathbb{E}^{\mathbb{Q}}\left[\tau(u)\right].\]

The existence of \(\mathbb{E}^{\mathbb{Q}}\left[S_i(\tau(u))\right]\) follows when: 1) \(\mathbb{E}^{\mathbb{Q}}\left[|S_i(1)|\right]<\infty\) (result by \cite{Doob1990} page 380, also used in Theorem~\ref{th_wald}), or, 2) there exists a real \(\vartheta<0\) such that \(\kappa_i^{\mathbb{Q}}(\vartheta)=\kappa_{\boldsymbol{S}}(\vartheta \boldsymbol{e}_i+\vartheta^*\boldsymbol{1})-\kappa_{\boldsymbol{S}}(\vartheta^*\boldsymbol{1})\leq 0\). Here we have used the notation \(\boldsymbol{e}_i\) to denote the vector of dimension \(d\) with all entries 0 except for the \(i\)-th entry which is 1. Note that in the special case where the L\'{e}vy processes \(S_1(\cdot),\ldots, S_d(\cdot)\) are independent of each other, condition 2) is satisfied when the individual risk process has a positive drift under the \(\mathbb{Q}\)-measure, i.e.\ \(\mathbb{E}^{\mathbb{Q}}[S_i(1)]=\kappa_i' (\vartheta^*)>0\).

In case \(\mathbb{E}^{\mathbb{Q}}\left[S_i(\tau(u))\right]\) has been proven to exist, it is given by 
\[\mathbb{E}^{\mathbb{Q}}[S_i(\tau(u))]=\mathbb{E}^{\mathbb{Q}}[S_i(1)]\,\mathbb{E}^{\mathbb{Q}}[\tau(u)].\]
When both \(\mathbb{E}^{\mathbb{Q}}[S(\tau(u))]\) and \(\mathbb{E}^{\mathbb{Q}}[S_i(\tau(u))]\) exist, expression \eqref{eq_th_wald} for an infinite time horizon holds true.

\vb

In the infinite time horizon case, it is also possible to derive an expression for the dynamic VaR measure and by differentiation obtain an expression for the gradient capital allocation.
By Theorem XI.2.3 of \cite{Asmussen2010}, the infinite time ruin probability is of an exponential form: \(\psi(u,\infty)=e^{-\vartheta^*u}\). The dynamic VaR measure is then determined as \({\rm VaR}^{\alpha}(S,\infty)=-({1}/{\vartheta^*})\,\ln(\alpha)\). Whenever the gradient allocation is properly defined, it is given by \[\frac{1}{(\vartheta^*)^2}\frac{\partial\vartheta^*(x_i)}{\partial x_i}\big|_{x_i=1},\] where \(\vartheta^*(x_i)\) denotes the `Cram\'{e}r root' of the process \(\sum_{j\neq i}S_j(\cdot)+x_iS_i(\cdot)\).

\subsection{Compound Poisson with drift}\label{ex_CP}
This subsection models the risk process by a compound Poisson model with drift as is popular in insurance risk modeling (the well known Cram\'{e}r-Lundberg model). We assume that the jumps of the individual processes are independent and identically exponentially distributed. The focus lies on finding expressions for the proposed capital allocation methods \(\rm{AVaR}_i^{\alpha}(S,\infty)\) and \(\overline{\rm{AVaR}}_i^{\alpha}(S,\infty)\) over an infinite time horizon. Furthermore, we compare these new capital allocation methods to the gradient allocation method when applied to the dynamic VaR measure, i.e.\ \(\rm{GVaR}_i^{\alpha}(S,\infty)\). Finally, we present some numerical work.

\vb

We start by specifying the model in more detail. Risk process \(i\) is independent of the other risk processes and defined as
\[S_i(t):=-r_it+\sum_{k=1}^{N_i(t)}Z_{i,k},\]
where the Poisson arrival process \(N_i(\cdot)\) is independent of the jump sizes \(Z_{i,k}\) and has rate \(\beta_i\), respectively. For risk process \(i\), the jump sizes \(Z_{i,k}\) are i.i.d. and exponentially distributed with parameter \(\theta\) and moment generating function \(\hat{F}_{Z}(\vartheta)=\theta/(\theta-\vartheta)\) which are the same for all \(i\). The Poisson processes \(N_i(\cdot)\) and jump size sequences \((Z_{i,k})_{k}\) are independent across \(i\). For risk process \(i\) the L\'{e}vy exponent is given by
\[\kappa_i(\vartheta):=-\vartheta r_i+\beta_i\left(\hat{F}_{Z}(\vartheta)-1\right)=-\vartheta r_i+\beta_i\frac{\vartheta}{\theta-\vartheta}.\]

By evaluation of the moment generating function, we will show that the aggregated risk process \(S(t):=\sum_{i=1}^d S_i(t)\) has constant fees \(r:=\sum_{i=1}^dr_i\), compound Poisson jumps with arrival rate \(\lambda:=\sum_{i=1}^d\beta_i\), and i.i.d.\ exponentially distributed jump sizes \(Z\) with parameter \(\theta\) (and m.g.f. \(\hat{F}_Z(\cdot)\)), i.e.
\[\mathbb{E}\left[\exp\left(\vartheta S(t)\right)\right]=\prod_{i=1}^d \exp\left(-r_it+\beta_i t(\hat{F}_Z(\vartheta)-1)\right)=\exp\left(-\vartheta r t + \lambda t(\hat{F}_Z(\vartheta)-1)\right).\]
This gives the L\'{e}vy exponent of the aggregated risk process \(S(\cdot)\) as \(\kappa(\vartheta)=-\vartheta r +\lambda \vartheta/(\theta-\vartheta)\).
By a similar argumentation we also find 
\[\mathbb{E}\left[\exp\left(\langle{\boldsymbol\vartheta},{\boldsymbol S}(t)\rangle\right)\right]=\exp\left(t\kappa_{\boldsymbol{S}}(\boldsymbol{\vartheta})\right)=\exp\left(-\sum_{i=1}^d \vartheta_i r_it+\lambda t\left(\sum_{i=1}^d\frac{\beta_i}{\lambda}\hat{F}_{Z}\left(\vartheta_i\right)-1\right)\right).\]

\vb 

To rule out the trivial situation where the ultimate ruin probability equals 1, we assume a negative drift, i.e.\ \(r>\frac{\lambda}{\theta}\). The negative drift assumption implies \(S(t)\rightarrow-\infty\) and \(\sup_{0\leq t<\infty}S(t)<\infty\) almost surely. Under this assumption the change of measure (to the \(\mathbb{Q}\)-measure) that was presented in Section~\ref{sec_Aallocation} can be applied. To this end, we take \(\vartheta^*\) as the positive solution to \(\kappa(\vartheta)=0\) which gives \(\vartheta^*=\theta-\frac{\lambda}{r}\) and find
\[m_i=\frac{\partial}{\partial \vartheta_i}\mathbb{E}\left[e^{\langle{\boldsymbol\vartheta},\boldsymbol{S}(1)\rangle}\right]\Bigg|_{{\boldsymbol\vartheta}=\vartheta^*{\boldsymbol 1}}=-r_i+\beta_i\frac{\theta}{(\theta-\vartheta^*)^2}=-r_i+\beta_i\theta\frac{r^2}{\lambda^2},\]
such that \(m=\sum_{i=1}^d m_i=-r+\theta\frac{r^2}{\lambda}\).
By Theorem~\ref{th_asymptotic}, \(\lim_{u\rightarrow \infty} c_i(,S,\infty)=m_i/m\) whenever \(\mathbb{E}^{\mathbb{Q}}\left[|S_i(1)|\right]<\infty\) for all \(i\). 

Under the  \(\mathbb{Q}\)-measure, we find, analogous to the computations in Section IV.4 of \cite{Asmussen2010} that the jump size $Z$ is again exponentially distributed with rate \(\theta^{\mathbb{Q}}=\theta-\vartheta^*=\frac{\lambda}{r}\) and that the jump arrivals are still Poisson distributed with parameter \(\beta_i^{\mathbb{Q}}=\beta_i\frac{\theta}{\theta-\vartheta^*}=\frac{\beta_i}{\lambda}\theta r\). As a consequence we also find \(\lambda^{\mathbb{Q}}=\theta r\). 

For the aggregated risk process \(S(\cdot)\) with exponential jump sizes, the infinite time ruin probability is known and given by (see Chapter IV, Section 5 in \cite{Asmussen2010}):
\begin{equation}\label{eq_ruin_CL}
\psi(u,\infty)=\frac{\lambda}{\theta r}e^{-\vartheta^* u}=\frac{\lambda}{\theta r}e^{-(\theta-\lambda/r) u}.
\end{equation}
From this we can extract the value at ruin, i.e.
\[\rm{VaR}^{\alpha}(S,\infty)=-\frac{1}{\vartheta^*}\ln \left(\frac{\alpha \theta r}{\lambda}\right).\] 
In the remainder of this section we will consider the three allocation methods discussed in this paper, i.e.
\begin{enumerate}
\item \(\rm{AVaR}_i^{\alpha}(S,\infty)\) (through \(K_i(u,S,\infty)\)), and,
\item \(\overline{\rm{AVaR}}_i^{\alpha}(S,\infty)\) (through \(\overline{K}_i(u,S,\infty)\)),
\item \(\rm{GVaR}_i^{\alpha}(S,\infty)\).
\end{enumerate}

To derive an expression for \(c_i(u,S,\infty)\) and \(K_i(u,S,\infty)\) (and subsequently \(\rm{AVaR}_i^{\alpha}(S,\infty)\)) for general \(u\) we, unfortunately, cannot use Wald's first identity as in Theorem~\ref{th_wald}. In order to derive an expression for these allocation quantities we will condition on the deficit at ruin and the time or ruin. For the numerator in the expression of \(c_i(u,S,\infty)\), as given in \eqref{eq_Kallocation}, we then find  
\begin{align*}
\mathbb{E}\left[S_i(\tau(u))\,\big|\,\tau(u)<\infty\right]=&\frac{1}{\psi(u,\infty)}\mathbb{E}^{\mathbb{Q}}\left[e^{-\vartheta^* S(\tau(u))}S_i(\tau(u))\right]\\
=&\frac{1}{\psi(u,\infty)}\int_0^\infty \int_u^\infty e^{-\vartheta^* x}\mathbb{E}^{\mathbb{Q}}\left[S_i(t)\,|\,S(t)=x,\tau(u)=t\right]f^{\mathbb{Q}}_{\tau(u),S(\tau(u))}(t,x)\,{\rm d}x\,{\rm d}t,
\end{align*}
where \(f^{\mathbb{Q}}_{\tau(u),S(\tau(u))}(t,x)\) denotes the joint density function of \(\tau(u)\) and \(S(\tau(u))\). The value of the process at the time of ruin can be written as \(S(\tau(u))=u+\xi(u)\), with overshoot \(\xi(u)\). The overshoot is exponentially distributed with parameter \(\theta\) (or \(\theta^{\mathbb{Q}}\) under \(\mathbb{Q}\)) and independent of the time of ruin (see also Proposition V.1.1 in \cite{Asmussen2010}). As a result we find 
\begin{equation}\label{eq_K_CP}
\mathbb{E}\left[S_i(\tau(u))\,\big|\,\tau(u)<\infty\right]=\frac{1}{\psi(u,\infty)}\int_0^\infty \int_u^\infty e^{-\vartheta^* x}\mathbb{E}^{\mathbb{Q}}\left[S_i(t)\,|\,S(t)=x\right]f^{\mathbb{Q}}_{\tau(u)}(t)f^{\mathbb{Q}}_{S(\tau(u))}(x)\,{\rm d}x\,{\rm d}t,
\end{equation}
where \(f^{\mathbb{Q}}_{\tau(u)}(t)\) denotes the probability density function of the time of ruin \(\tau(u)\) and \(f^{\mathbb{Q}}_{S(\tau(u))}(x)\) denotes the probability density function of \(S(\tau(u))\).
Next, we find that the conditional expectation \(\mathbb{E}^{\mathbb{Q}}\left[S_i(t)\,|\,S(t)=x\right]\) can be derived explicitly by noting that, for fixed \(n_i\) and \(n_k\), 
\[\mathbb{E}^{\mathbb{Q}}\left[\sum_{j_1=1}^{n_i}Z_{i,j_1}\,\Bigg|\,\sum_{j_1=1}^{n_i}Z_{i,j_1}+\sum_{j_2=1}^{n_k}Z_{k,j_2}=y\right]=\frac{n_i}{n_i+n_k}y.\] After some tedious but straightforward calculations this gives 
\[\mathbb{E}^{\mathbb{Q}}\left[S_i(t)\,|\,S(t)=x\right]=\frac{\beta_i^{\mathbb{Q}}}{\sum_{j=1}^d\beta_j^{\mathbb{Q}}}\left(x+rt\right)-r_it.\] 
By substituting this result into Equation \eqref{eq_K_CP}, we find
\begin{align}
\mathbb{E}\left[S_i(\tau(u))\,\big|\,\tau(u)<\infty\right]=&\frac{1}{\psi(u,\infty)}\frac{\beta_i^{\mathbb{Q}}}{\sum_{j=1}^d\beta_j^{\mathbb{Q}}}\int_u^\infty \int_0^\infty e^{-\vartheta^*x}x f^{\mathbb{Q}}_{\tau(u)}(t)f^{\mathbb{Q}}_{S(\tau(u))}(x)\,{\rm d}t\,{\rm d}x\nonumber\\
&+\frac{1}{\psi(u,\infty)}\left(\frac{\beta_i^{\mathbb{Q}}}{\sum_{j=1}^d\beta_j^{\mathbb{Q}}}r-r_i\right)\int_u^\infty \int_0^\infty e^{-\vartheta^*x}t f^{\mathbb{Q}}_{\tau(u)}(t)f^{\mathbb{Q}}_{S(\tau(u))}(x)\,{\rm d}t\,{\rm d}x\label{eq_K_CP_int}
\end{align}
We will now discuss the two double integrals separately. 
The first double integral can be written as 
\small
\begin{equation}\label{eq_K_CP_int1}
\mathbb{E}^{\mathbb{Q}}\left[S(\tau(u))e^{-\vartheta^* S(\tau(u))}\right]=-\frac{\partial}{\partial \vartheta^*}\mathbb{E}^{\mathbb{Q}}\left[e^{-\vartheta^* S(\tau(u))}\right]=-\frac{\partial}{\partial \vartheta^*}e^{-\vartheta^* u} \frac{\theta^\mathbb{Q}}{\theta^\mathbb{Q}+\vartheta^*}=e^{-\vartheta^* u}\frac{\theta^\mathbb{Q}}{\theta^\mathbb{Q}+\vartheta^*}\left(u+\frac{1}{\theta^\mathbb{Q}+\vartheta^*}\right),
\end{equation}
\normalsize
where we have used that \(S(\tau(u))=u+\xi(u)\), with overshoot \(\xi(u)\) exponentially distributed with parameter \(\theta^{\mathbb{Q}}\) under \(\mathbb{Q}\) and Laplace transform \(\mathbb{E}^{\mathbb{Q}}\left[e^{-\vartheta^* \xi(\tau(u))}\right]=\theta^{\mathbb{Q}}/(\theta^{\mathbb{Q}}+\vartheta^*)\).

The second double integral can be written as
\begin{equation}\label{eq_K_CP_int2}
\mathbb{E}^{\mathbb{Q}}\left[\tau(u)e^{-\vartheta^* S(\tau(u))}\right]=e^{-\vartheta^* u}\mathbb{E}^{\mathbb{Q}}\left[\tau(u)\right]\mathbb{E}^{\mathbb{Q}}\left[e^{-\vartheta^* \xi(u)}\right]=\mathbb{E}^{\mathbb{Q}}\left[\tau(u)\right]e^{-\vartheta^* u} \frac{\theta^\mathbb{Q}}{\theta^\mathbb{Q}+\vartheta^*}.
\end{equation}
As a consequence of Wald's identity we furthermore have that 
\begin{equation}
\mathbb{E}^{\mathbb{Q}}\left[\tau(u)\right]=\frac{\mathbb{E}^{\mathbb{Q}}\left[S(\tau(u))\right]}{\mathbb{E}^{\mathbb{Q}}\left[S(1)\right]}=\frac{u+1/\theta^\mathbb{Q}}{m}.\label{eq_CP_tau}
\end{equation}
Substituting \eqref{eq_K_CP_int1}, \eqref{eq_K_CP_int2} and \eqref{eq_CP_tau} back into equation \eqref{eq_K_CP_int} gives
\[\mathbb{E}\left[S_i(\tau(u))\,\big|\,\tau(u)<\infty\right]=\frac{e^{-\vartheta^* u}}{\psi(u,\infty)}\frac{\theta^\mathbb{Q}}{\theta^\mathbb{Q}+\vartheta^*}\left(\frac{\beta_i^{\mathbb{Q}}}{\sum_{j=1}^d\beta_j^{\mathbb{Q}}}\left(u+\frac{1}{\theta^\mathbb{Q}+\vartheta^*}\right)+\left(\frac{\beta_i^{\mathbb{Q}}}{\sum_{j=1}^d\beta_j^{\mathbb{Q}}}r-r_i\right)\frac{u+1/\theta^\mathbb{Q}}{m} \right).\]
Substituting the known expression for the ultimate ruin probability \eqref{eq_ruin_CL}, we finally find 
\[\mathbb{E}\left[S_i(\tau(u))\,\big|\,\tau(u)<\infty\right]=\frac{\beta_i^{\mathbb{Q}}}{\sum_{j=1}^d\beta_j^{\mathbb{Q}}}\left(u+\frac{1}{\theta^\mathbb{Q}+\vartheta^*}\right)+\left(\frac{\beta_i^{\mathbb{Q}}}{\sum_{j=1}^d\beta_j^{\mathbb{Q}}}r-r_i\right)\frac{u+1/\theta^\mathbb{Q}}{m}.\]
By summation over \(i\) we get \(\mathbb{E}\left[S(\tau(u))\,\big|\,\tau(u)<\infty\right]=u+1/\theta\) and by the definition of \(c_i(u,S,\infty)\) we find 
\[c_i(u,S,\infty)=\frac{\beta_i^{\mathbb{Q}}}{\sum_{j=1}^d\beta_j^{\mathbb{Q}}}+\left(\frac{\beta_i^{\mathbb{Q}}}{\sum_{j=1}^d\beta_j^{\mathbb{Q}}}r-r_i\right)\frac{u+1/\theta^\mathbb{Q}}{m(u+1/\theta)}, \ K_i(u,S,\infty)=c_i(u,S,\infty)u.\]
For \(u\rightarrow\infty\) this expression coincides with \(m_i/m\) as has been proven in Theorem~\ref{th_asymptotic}.

\vb

By similar argumentation we can also derive an expression for the allocations \(\overline{K}_i(u,S,\infty)\) given in Equation \eqref{eq_K_tilde_allocation}. By the negative drift assumption, the supremum of the process \(S(\cdot)\) is almost surely finite. Analogue to the derivation of \eqref{eq_t^*_tau}, we can rewrite the expression of \(\overline{K}_i(u,S,\infty)\) which is dependent on \(t_\infty^*\) in terms of \(\tau(u)\), i.e.
\begin{align*}
 \overline{K}_i(u,S,\infty)=&\mathbb{E}\left[S_i(t_\infty^*)\,\big|\,S(t_\infty^*)=u, t_\infty^*<\infty\right]=\mathbb{E}\left[S_i(\tau(u))\,\bigg|\,S(\tau(u))=u,\tau(u)<\infty,\sup_{\tau(u)\leq t<\infty}S(t)=u\right]\\
=&\mathbb{E}\left[S_i(\tau(u))\,\big|\,S(\tau(u))=u,\tau(u)<\infty\right]=\frac{\int_0^\infty xe^{-\vartheta^*u}f^{\mathbb{Q}}_{S_i(\tau(u)),S(\tau(u))}(x,u){\rm d}x}{e^{-\vartheta^*u}f^{\mathbb{Q}}_{S(\tau(u))}(u)}\\
=&\mathbb{E}^{\mathbb{Q}}\left[S_i(\tau(u))\,\big|\,S(\tau(u))=u\right]=\int_{0}^\infty \mathbb{E}^{\mathbb{Q}}\left[S_i(t)\,|\,S(t)=u, \tau(u)=t\right]\,f^{\mathbb{Q}}_{\tau(u)\,|\,S(\tau(u))}(t\,|\,u)\,{\rm d}t\\
=&\int_{0}^\infty \mathbb{E}^{\mathbb{Q}}\left[S_i(t)\,|\,S(t)=u\right]\,f^{\mathbb{Q}}_{\tau(u)}(t)\,{\rm d}t.=\frac{\beta_i^{\mathbb{Q}}}{\sum_{j=1}^d\beta_j^{\mathbb{Q}}}u+\left(\frac{\beta_i^{\mathbb{Q}}}{\sum_{j=1}^d\beta_j^{\mathbb{Q}}}r-r_i\right)\int_{0}^\infty t\,f^{\mathbb{Q}}_{\tau(u)}(t)\,{\rm d}t\\
=& \frac{\beta_i^{\mathbb{Q}}}{\sum_{j=1}^d\beta_j^{\mathbb{Q}}}u+\left(\frac{\beta_i^{\mathbb{Q}}}{\sum_{j=1}^d\beta_j^{\mathbb{Q}}}r-r_i\right)\frac{u+1/\theta^{\mathbb{Q}}}{m},
\end{align*}
where \(f^{\mathbb{Q}}_{S_i(\tau(u)),S(\tau(u))}(x,y)\) denotes the joint probability density function of \(S_i(\tau(u))\) and \(S(\tau(u))\).

Unlike \(K_i(u,S,\infty)\), the allocations \(\overline{K}_i(u,S,\infty)\) do sum up to \(u\) as expected. One should also note that \(\overline{c}_i(u,S,\infty)\) converges to \(c_i(u,S,\infty)\) (or equivalently \(m_i/m\)) for \(u\rightarrow \infty\).

\vb

The gradient capital allocations can be derived by differentiation of the ruin probability as mentioned in Section~\ref{sec_GVaR} or by use of Theorem~\ref{th_Euler}. With respect to the former, note that the aggregated process \(\sum_{j\neq i}^d S_j(t)+x_iS_i(t)\) no longer has exponential claims but phase-type \({\rm PH}(\boldsymbol{\gamma} , \boldsymbol{M}(x_i))\) distributed claims, where we have used the same notation as in \cite{Drekic2004} with
\begin{align*}
\boldsymbol{\gamma} = (\gamma_1,\ldots,\gamma_{d})=\left(\frac{\beta_1}{\lambda},\ldots,\frac{\beta_d}{\lambda}\right), \ \ \ \ \ \boldsymbol{M}(x_i)={\rm diag}\{-\theta,\ldots,-\theta/x_i,\ldots,-\theta\}. 
\end{align*}
The infinite time ruin probability can be found by performing a number of matrix operations (see Chapter IX, Section 3 in \cite{Asmussen2010}), i.e.
\begin{equation}
\mathbb{P}\left(\sup_{t\in[0,\infty)}\sum_{j\neq i}^d S_j(t)+x_iS_i(t)\geq u\right)=\boldsymbol{\gamma}_+(x_i)e^{\left(\boldsymbol{M}(x_i)-\boldsymbol{M}(x_i)\boldsymbol{e}\boldsymbol{\gamma}_+\right)u}\boldsymbol{e}, \ \ \ \ \boldsymbol{\gamma}_+(x_i)=-\frac{\lambda}{r} \boldsymbol{\gamma M}(x_i)^{-1},\label{eq_psi_phase}
\end{equation}
where \(\boldsymbol{e}\) is the column vector of length \(d\) with all components equal to one.
By differentiation of \eqref{eq_psi_phase} (as mentioned in Section~\ref{sec_GVaR}), the gradient allocations \(\rm{GVaR}_i^{\alpha}(S,\infty)\) can be found. 

\vb

We note that this example is also captured under Theorem~\ref{th_Euler}.  In the next numerical section, we show that the gradient allocation method coincides with  \(\overline{\rm{AVaR}}_i^{\alpha}(S,T)\) (the result of Theorem~\ref{th_Euler}) even on an infinite time horizon.

\subsubsection{Numerical Example}\label{subsub_CP_num}
For the numerical results and comparison between the different allocation methods we consider the two-dimensional case \((d=2)\) and use a setup that aligns with the one considered in \cite{Asmussen1984}.

\begin{itemize}\itemsep0em
\item[$\circ$] We consider the case that the jump sizes are exponentially distributed with parameter \(\theta=1\).
\item[$\circ$] The drift rates are given by \(r_{1,2}\equiv r=1\).
\item[$\circ$] The individual jump intensities are given by \(\beta_1=0.85\) and \(\beta_2=0.95\).
\end{itemize}

With these parameter settings, the negative drift assumption of the aggregated risk process \(S(\cdot)\) is satisfied. In Figure \ref{fig3}(b) we present the allocation fractions \({{\rm GVaR}^\alpha_1(S,\infty)}/{{\rm VaR}^\alpha(S,\infty)}\), \({\overline{\rm VaR}^\alpha(S,\infty)}/{{\rm VaR}^\alpha(S,\infty)}\), and \({{\rm AVaR}^\alpha_1(S,\infty)}/{{\rm VaR}^\alpha(S,\infty)}\)  for the first risk process (or business line) \(S_1(\cdot)\) as a function of \(u\). The allocation fraction \(\overline{c}_1(u,S,\infty)\) can be seen to converge to \(c_1(u,S,\infty)\) (and thus also \(m_i/m\)) as \(u\) becomes large. Figure \ref{fig3}(a) shows the same convergence for \(\alpha\rightarrow 0\) when considering the allocation of the measure \({\rm VaR}^{\alpha}(S,\infty)\). This figure also illustrates that, similar to the Brownian case, \({\rm GVaR}_1^{\alpha}(S,\infty)/{\rm VaR}^{\alpha}(S,\infty)\) and \(\overline{\rm AVaR}_1^{\alpha}(S,\infty)/{\rm VaR}^\alpha(S,\infty)\) coincide.

\begin{figure}[h]
 \centering
 {\includegraphics[width=7.6cm]{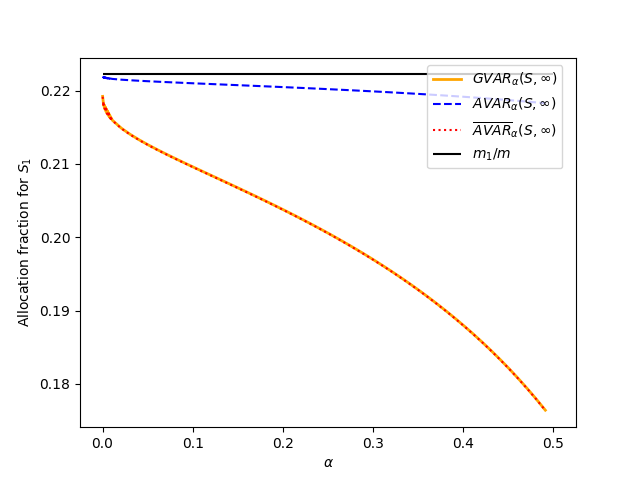}}\hspace{1cm}%
{\includegraphics[width=7.6cm]{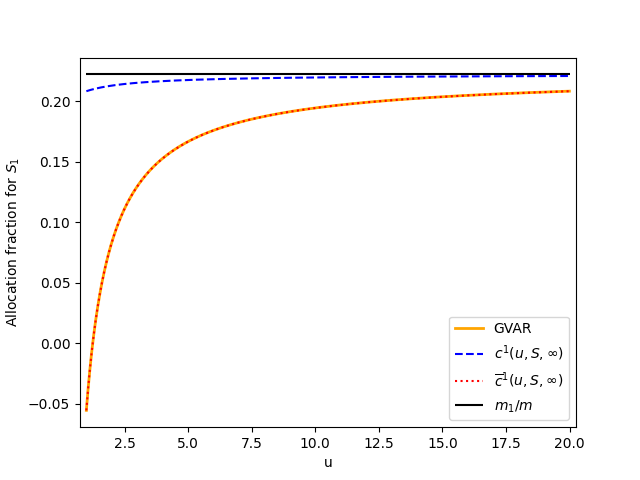}}
\caption{Allocation fractions of the first risk process as a function of \(\alpha\) and \(u\) for $ \beta_1=0.85,\beta_2=0.95, \theta=1, r_1=1,r_2=1$. Panel (a), on the left: 
 \({{\rm GVaR}^\alpha_1(S,\infty)}/{{\rm VaR}^\alpha(S,\infty)}\), \({\overline{\rm VaR}^\alpha(S,\infty)}/{{\rm VaR}^\alpha(S,\infty)}/{{\rm VaR}^\alpha(S,\infty)}\), and \({{\rm AVaR}^\alpha_1(S,\infty)}/{{\rm VaR}^\alpha(S,\infty)}\) as functions of \(\alpha\). Panel (b), on the right: \(c_1(u,S,\infty)\) and \(\overline{c}_1(u,S,\infty)\) as a function of \(u\).}
\label{fig3}
\end{figure}

Figure \ref{fig4} presents the allocation fractions \({{\rm GVaR}^\alpha_1(S,\infty)}/{{\rm VaR}^\alpha(S,\infty)}\), \({{\rm AVaR}^\alpha_1(S,\infty)}/{{\rm VaR}^\alpha(S,\infty)}\)  

and \({\overline{\rm AVaR}^\alpha_1(S,\infty)}/{{\rm VaR}^\alpha(S,\infty)}\), and their sensitivity towards some of the underlying parameters. When both risk processes become less risky (see Figure \ref{fig4}(a)), the allocation fractions move towards a more even risk distribution. The current parameter setup also shows a relatively high impact of a change in the jump intensities. In Figure \ref{fig4}(c), the jump intensities are adjusted favorably for the first risk process resulting in negative risk/capital allocations. 

\begin{figure}[h]
 \centering
 {\includegraphics[width=7.6cm]{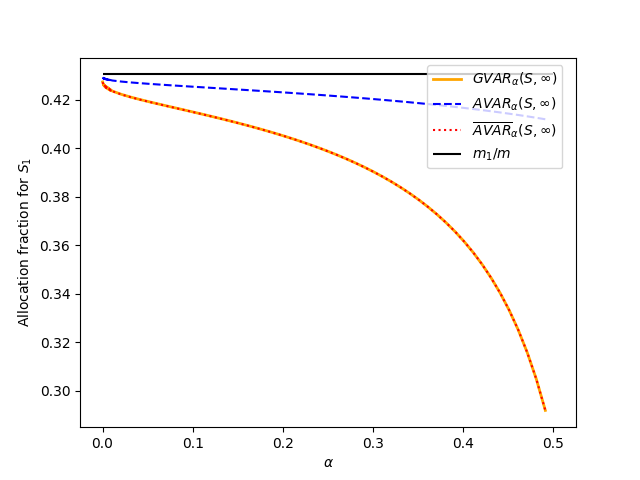}}\hspace{1cm}%
{\includegraphics[width=7.6cm]{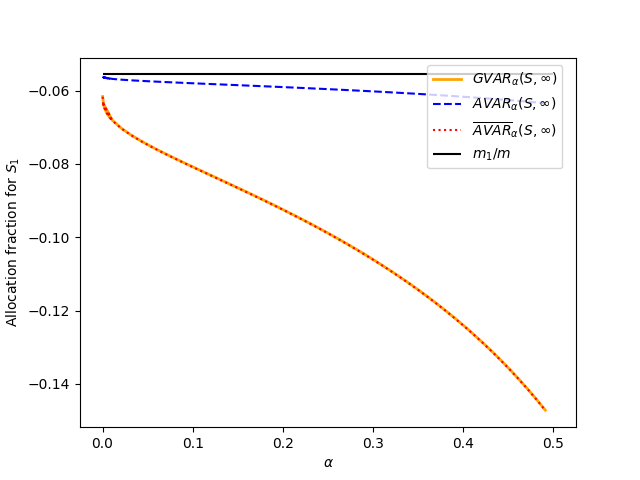}}
\caption{Allocation fractions of the first risk process \({{\rm AVaR}^\alpha_1(S,\infty)}/{{\rm VaR}^\alpha(S,\infty)}\),  \({\overline{\rm AVaR}^\alpha_1(S,\infty)}/{{\rm VaR}^\alpha(S,\infty)}\) and \({{\rm GVaR}^\alpha_1(S,\infty)}/{{\rm VaR}^\alpha(S,\infty)}\) as a function of \(\alpha\). Panel (a), on the left: $\beta_1=0.85,\beta_2=0.95, \theta=1, r_1=1.5, r_2=1.5$. Panel (b), on the right: $ \beta_1=0.8,\beta_2=1, \theta=1, r_1=1,r_2=1$.}
\label{fig4}
\end{figure}

\section{Concluding remarks}\label{sec_conclusion}
This paper has addressed methodologies to allocate capital reserves to multiple risk process (to be thought of as e.g., business lines). We introduced an intuitively appealing, novel allocation method, with a focus on its application to capital reserves which are determined through a dynamic VaR type measure. Various desirable properties of the presented approach were derived including a limit result when considering a large time horizon and the comparison with the frequently used gradient allocation method. In passing we introduced a second allocation method, and discussed its relation to the other allocation approaches. A number of examples illustrated the applicability and performance of the allocation approaches. 

\vb

Theorem~\ref{th_Euler}, featuring the gradient allocation method applied to the dynamic VaR measure, has been tailored to our needs and captures the examples given in Section~\ref{ex_BM} \& \ref{ex_CP}. One could further investigate whether an extension or adjustment of Theorem~\ref{th_Euler} can be made to include more risk processes.  This requires a different approach as the current result and proof require the maximum of the aggregated process to be obtained. Furthermore, the current proof relies on the differentiability (with respect to an individual risk process) of the sample path of the maximum aggregated process. 

Follow-up research could also relate to necessary and sufficient conditions for diversification and concentration properties of the allocated risk measures. Examples include the `no undercut' property, which has been established for the Brownian case when considering an infinite time horizon. For a finite time horizon and other risk processes these types of properties have not been dealt with in this paper. 

\bibliography{References}

\appendix
\begin{appendices}
\section{Supporting Results}\label{app_A}
The following result is Theorem 9.4 in \cite{Loomis1990}. Note that it is a slightly different version of the classical Implicit Function Theorem.
\begin{theorem}{\bf (Implicit Function Theorem)}\label{th_app1} 
Let \(X\times P\) be an open subset of \(\mathbb{R}\times \mathbb{R}\) and let \(f: X\times P \rightarrow \mathbb{R}\) be differentiable. Suppose the derivative \(D_xf\) of \(f\) with respect to \(x\) is continuous on \(X\times P\). Assume that \(D_x f(\overline{x},\overline{p})\) is invertible where the point \((\overline{x},\overline{p})\) in the interior of \(X\times P\). Let \[\overline{y}=f(\overline{x},\overline{p}).\]
Then there are neighborhoods \(U\subset X\) and \(W\subset P\) of \(\overline{x}\) and \(\overline{p}\) on which \(f(x,p)=y\) uniquely defines \(x\) as a function of \(p\). That is, there is a function \(\xi: W\rightarrow U\) such that:
\begin{enumerate}[a.]
\item \(f(\xi(p),p)=\overline{y}\) for all \(p\in W\).
\item For each \(p\in W, \ \epsilon(p)\) is the unique solution to \(f(x,p)=y\) lying in \(U\). In particular, then 
\[\xi(\overline{p})=\overline{x}.\]
\item \(\xi\) is differentiable on \(W\), and 
\[\frac{\partial\xi}{\partial p}=-\left(\frac{\partial f}{\partial x}\right)^{-1}\frac{\partial f}{\partial p}\]
\end{enumerate}
\end{theorem}

\begin{lemma}\label{lem_bi_BM}
Consider the bivariate normal distribution \[
\begin{pmatrix}
X_1\\
X_2
\end{pmatrix} 
\sim \mathcal{N}\left(\begin{pmatrix}
\mu_1\\ 
\mu_2
\end{pmatrix},\begin{pmatrix}
\sigma^2_1 & \rho\sigma_1\sigma_2\\
\rho\sigma_1\sigma_2 & \sigma_2^2
\end{pmatrix}\right).\]
The conditional distribution of \(|X_1|\) given \(X_2=x_2\) is 
\[\mathbb{E}\left[|X_1| \big| X_2=x_2\right]=\left(\mu_1+\rho\frac{\sigma_1}{\sigma_2}\left(x_2-\mu_2\right)\right)\left(1-2\Phi\left(c\right)\right)+2\sigma_1\sqrt{1-\rho^2}\phi\left(c\right),\]
where \(c=\frac{-\mu_1-\rho\frac{\sigma_1}{\sigma_2}(x_2-\mu_2)}{\sigma_1\sqrt{1-\rho^2}}\).
\end{lemma}
\begin{proof}
First note that we can write \(X_1 = \mu_1+\sigma_1\left(\rho\frac{X_2-\mu_2}{\sigma_2}+\sqrt{1-\rho^2}Z\right)\), where \(Z\) is a standard normal random variable independent of \(X_2\). This gives
\[\mathbb{E}\left[|X_1| \big| X_2=x_2\right]=\mathbb{E}\left[\bigg|\mu_1+\sigma_1\left(\rho\frac{x_2-\mu_2}{\sigma_2}+\sqrt{1-\rho^2}Z\right)\bigg| \right].\]
Conditioning on the events \(X_1>0\) and \(X_1\leq 0\) we find 
\begin{align*}
\mathbb{E}\left[|X_1| \big| X_2=x_2\right]=&\left(\mu_1+\rho\frac{\sigma_1}{\sigma_2}(x_2-\mu_2)\right)\left(\mathbb{P}(Z>c)-\mathbb{P}(Z\leq c)\right)\\
&+\sigma_1\sqrt{1-\rho^2}\left(\mathbb{E}\left[Z\big|Z>c\right]\mathbb{P}(Z>c)-\mathbb{E}\left[Z\big|Z\leq c\right]\mathbb{P}(Z\leq c)\right)\\
=&\left(\mu_1+\rho\frac{\sigma_1}{\sigma_2}(x_2-\mu_2)\right)\left(1-2\Phi(c)\right)+2\sigma_1\sqrt{1-\rho^2}\phi(c)
\end{align*}
where we have used that \(\mathbb{E}\left[Z\big|Z\leq c\right]\mathbb{P}(Z\leq c)=\int_{-\infty}^cz\phi(z)dz=-\int_{-\infty}^c\phi'(z)dz=-\phi(c)\) and similarly \(\mathbb{E}\left[Z\big|Z> c\right]\mathbb{P}(Z> c)=\phi(c)\).
\end{proof}

\section{Proof of Theorem~\ref{th_Euler}}\label{app_B}
\begin{proof}
Without loss of generality we will prove the result for \(i=1\). We define the function \(F_1(y_1,x_1)\) for \(x_1\in(1-\delta,1+\delta)\) and \(y_1\in(q_{{\rm VaR}^{\alpha}_{1,T}}(x_1)-\bar{\delta}/2,q_{{\rm VaR}^{\alpha}_{1,T}}(x_1)+\bar{\delta}/2)\) by 
\[F_1(y_1,x_1):=\mathbb{P}(Z(x_1)\leq y_1)=\mathbb{E}\left[\mathbbm{1}_{Z(x_1)\leq y_1}\right],\]
where we have omitted the dependence of \(Z_{1,T}(x_1)\) on \(1,T\). 

First, we show that the function \(F_1(y_1,x_1)\) is: 1) continuously differentiable in \(y_1\), and, 2) differentiable in \(x_1\) for \(x_1\in(1-\delta,1+\delta)\) and \(y_1\in(q_{{\rm VaR}^{\alpha}_{1,T}}(x_1)-\bar{\delta}/2,q_{{\rm VaR}^{\alpha}_{1,T}}(x_1)+\bar{\delta}/2)\). 
\begin{enumerate}[1)]
\item To prove the continuous differentiability with respect to \(y_1\) we note that \(\frac{\partial F_1}{\partial y_1}(y_1,x_1)=f_{1,x_1}(y_1)\), which is assumed to be continuous in \(x_1\) and \(y_1\) in the given interval by assumption (i). Furthermore, by the same assumption, it is continuous on a closed bounded interval and thereby finite.
\item To prove the differentiability of \(F_1(y_1,x_1)\) with respect to \(x_1\) we will approximate the discontinuous indicator function with a smoother function \(g\), which for \(\epsilon\) small enough such that \(0<\epsilon\leq \bar{\delta}/2\), is given by

\[g_{\epsilon,y_1}(z)=\begin{cases} 
1, & \text{for } y_1-z>\epsilon\\
\frac{1}{2}+\frac{y_1-z}{2\epsilon}, & -\epsilon\leq y_1-z\leq \epsilon\\
0, & \text{for } y_1-z<-\epsilon
\end{cases}.
\]
Note that the derivative is given by
\[g_{\epsilon,y_1}'(z)=\begin{cases} 
0, & \text{for } |y_1-z|> \epsilon\\
-\frac{1}{2\epsilon}, & \text{for } |y_1-z|< \epsilon
\end{cases},
\]
where the derivative does not exist when \(|y_1-z|=\epsilon\). For the interval of \(z\) where then derivative exists and is non-zero we have \((y_1-\epsilon,y_1+\epsilon)\subseteq (q_{{\rm VaR}^{\alpha}_{1,T}}(x_1)-\bar{\delta},q_{{\rm VaR}^{\alpha}_{1,T}}(x_1)+\bar{\delta})\). As a result of the continuity of \(Z(x_1)\), as in (i), the probability that \(|y_1-Z(x_1)|=\epsilon\) is zero. We will focus on showing that \(\mathbb{E}[g_{\epsilon,y_1}(Z(x_1))]\) is differentiable in \(x_1\) and that its derivatives may be computed by taking the derivative inside the expectation. 
To prove this we invoke the dominated convergence theorem. To this end, we first establish the differentiability of the sample paths of \(g_{\epsilon,y_1}(Z(x_1))\). Using similar notation as in \cite{Bertoin1996} and \cite{Sato1999}, we note that \(\frac{\partial}{\partial x_1}g_{\epsilon,y_1}(Z(x_1,\omega))\) for fixed \(\omega\in\Omega\) exists almost everywhere (except when \(|Z(x_1,\omega)-y_1|= \epsilon\) but this event is of probability zero). Furthermore, when the derivative exists it is given by
\[\frac{\partial}{\partial x_1}g_{\epsilon,y_1}(Z(x_1,\omega))=g_{\epsilon,y_1}'\left(Z(x_1,\omega)\right)\frac{\partial}{\partial x_1}Z(x_1,\omega).\]
The differentiability of \(Z(x_1,\omega)\) w.r.t.\ \(x_1\) can be obtained using Lemma~\ref{lem_der_fun}. 
L\'{e}vy processes excluding compound Poisson processes (without drift), almost surely obtain the supremum over a finite time horizon at a unique point in time (see page 171 in \cite{Kyprianou2006}).  Continuous L\'{e}vy processes as well as compound Poisson processes with non-zero drift and positive jumps both attain their supremum, i.e. \ the supremum is in fact a maximum.   As a result, we have \(A_{1,T}^*(x_1,\omega)\) non-empty and a singleton. In other words, the supremum is uniquely attained in \(t_{1,T}^*(x_1,\omega)\). 
We will now show that \(Z(x_1,\omega)\) is differentiable w.r.t. \(x_1\) with \(\frac{\partial}{\partial x_1}Z(x_1,\omega)=S_1(t_{1,T}^*(x_1),\omega)\) by making use of Lemma~\ref{lem_der_fun} and considering the two instances of \(S_i(\cdot)\) separately:  1) continuous processes, and, 2) compound Poisson processes with  negative drift and positive jumps. First note that, the function \(Z(x_1,\omega)\) maximizes over is of the form \(p(t)+x_1q(t)\), where \(p(t)=\sum_{j\neq 1}^dS_j(t,\omega)\) and \(q(t)=S_1(t,\omega)\). 
\begin{enumerate}
\item In case the processes \(S_i(\cdot)\) have continuous sample paths then \(p(t)\) and \(q(t)\) are continuous functions and by Lemma~\ref{lem_der_fun} we have \(\frac{\partial}{\partial x_1}Z(x_1,\omega)=S_1(t_{1,T}^*(x_1),\omega)\).
\item For compound Poisson processes it is well-known that over a finite interval the number of jumps is also almost surely finite, this property is often referred to as finite activity. As a result,  the compound Poisson process \(\sum_{j\neq 1}^dS_j(\cdot)+x_1 S_1(\cdot)\) with negative drift and positive jumps can only attain its maximum at a finite number of time points almost surely.  These time points coincide with the jump times of the individual compound Poisson processes \(S_i(\cdot)\). As these jump times do not depend on \(x_1\),  we have \(\frac{\partial}{\partial x_1}Z(x_1,\omega)=S_1(t_{1,T}^*(x_1),\omega)\) by Lemma~\ref{lem_der_fun}.
\end{enumerate}
We conclude that \(\frac{\partial}{\partial x_1}Z(x_1,\omega)=S_1(t_{1,T}^*(x_1),\omega)\) almost surely and almost everywhere (excluding the points where \(|Z(x_1,\omega)-y_1|=\epsilon\)), 
\(g_{\epsilon,y_1}'(Z(x_1,\omega))=-\frac{1}{2\epsilon}\mathbbm{1}_{|Z(x_1,\omega)-y_1|\leq \epsilon}\). This gives almost surely,
\[\frac{\partial}{\partial x_1}g_{\epsilon,y_1}(Z(x_1))=-\frac{1}{2\epsilon}\mathbbm{1}_{|Z(x_1)-y_1|<\epsilon}S_1(t^*_{1,T}(x_1)).\] 

Note that we always have
\[\bigg|\frac{g_{\epsilon,y_1}(Z(x_1+h,\omega))-g_{\epsilon,y_1}(Z(x_1,\omega))}{h}\bigg|\leq \frac{1}{2\epsilon}\sup_{0\leq t\leq T}\big| S_1(t,\omega)\big|,\]
where the majorizing function 
does not depend on \(x_1\) and its expectation \(\mathbb{E}\left[\frac{1}{2\epsilon}\sup_{0\leq t\leq T}\big| S_1(t)\big|\right]\) is finite by assumption (iv) and the fact that \(\epsilon>0\).
Hence, using the dominated convergence theorem to interchange the expectation and the limit, we have
\begin{align*}
\frac{\partial}{\partial x_1}\mathbb{E}[g_{\epsilon,y_1}(Z(x_1))]&=\lim_{h\rightarrow 0}\mathbb{E}\left[\frac{g_{\epsilon,y_1}(Z(x_1+h,\omega))-g_{\epsilon,y_1}(Z(x_1,\omega))}{h}\right]\\
&=\mathbb{E}\left[\lim_{h\rightarrow 0}\frac{g_{\epsilon,y_1}(Z(x_1+h,\omega))-g_{\epsilon,y_1}(Z(x_1,\omega))}{h}\right]=\mathbb{E}\left[\frac{\partial}{\partial x_1}g_{\epsilon,y_1}(Z(x_1))\right].
\end{align*}

Conditioning on the supremum process then gives
\begin{align*}
\frac{\partial}{\partial x_1}\mathbb{E}\left[g_{\epsilon,y_1}(Z(x_1))\right]&=-\int_{y_1-\epsilon}^{y_1+\epsilon}\frac{1}{2\epsilon}\mathbb{E}\left[S_1(t_{1,T}^*(x_1))|Z(x_1)=z\right]f_{1,x_1}(z)dz\\
&\xrightarrow{\epsilon\rightarrow 0} -\mathbb{E}\left[S_1(t_{1,T}^*(x_1))|Z(x_1)=y_1\right]f_{1,x_1}(y_1),
\end{align*}
where the limit follows from the fundamental theorem of calculus by noting that the expression inside the integral is continuous in by assumptions (i) and (iii).

\vb

Introducing the notation \(l(x_1):=-\mathbb{E}\left[S_1(t_{1,T}^*(x_1))|Z(x_1)=y_1\right]f_{1,x_1}(y_1)\), we will continue to show that 
\(\frac{\partial F_1(y_1,x_1)}{\partial x_1}=l(x_1)\). Using the new notation we have already shown that \(\frac{\partial}{\partial x_1}\mathbb{E}\left[g_{\epsilon,y_1}(Z(x_1))\right]\xrightarrow{\epsilon\rightarrow 0}l(x_1)\). By integration (of \(x_1\)) we would like to retrieve an expression for \(\mathbb{E}\left[g_{\epsilon,y_1}(Z(x_1))\right]\). In order to do so, we will interchange the integral (from 0 to \(x_1\)) and the limit (\(\epsilon\rightarrow 0\)). To this end, note that
\begin{align*}
\Bigg|\frac{\partial}{\partial x_1}\mathbb{E}\left[g_{\epsilon,y_1}(Z(x_1))\right]\Bigg|&\leq \mathbb{E}\left[\bigg|\frac{\partial}{\partial x_1}g_{\epsilon,y_1}(Z(x_1))\bigg|\right]=\int_{y_1-\epsilon}^{y_1+\epsilon}\frac{1}{2\epsilon}\mathbb{E}\left[\big|S_1(t_{1,T}^*(x_1))\big|\bigg|Z(x_1)=z\right]f_{1,x_1}(z)dz\\
& \leq M_1M_2<\infty.
\end{align*}
Here we have used that by the continuity of \(f_{1,x_1}(z)\) on the bounded interval \(x_1\in [1-\delta,1+\delta]\) and \(z\in [q_{{\rm VaR}^{\alpha}_{1,T}}(x_1)-\bar{\delta},q_{{\rm VaR}^{\alpha}_{1,T}}(x_1)+\bar{\delta}]\) by assumption (i), there exists some finite \(M_1\) independent of \(x_1\) and \(z\) such that \(f_{1,x_1}(z)\leq M_1\) on the same interval. Similarly we find \(M_2\) (also independent of \(x_1\) and \(z\)) as a bound for \(\mathbb{E}\left[|S_1(t_{1,T}^*(x_1))|\big|Z(x_1)=z\right]\) by assumption (iii).
As a result, we have shown that \(\big|\frac{\partial}{\partial x_1}\mathbb{E}\left[g_{\epsilon,y_1}(Z(x_1))\right]\big|\) is dominated by some finite constant \(M_1M_2\) independent of \(\epsilon\). Invoking the dominated convergence theorem, we interchange the integral (from 0 to \(x_1\)) and the limit \((\epsilon\rightarrow 0)\) and find for some constant \(c\),
\[\mathbb{E}\left[g_{\epsilon,y_1}(Z(x_1))\right]\xrightarrow{\epsilon\rightarrow 0}c +\int_{0}^{x_1}l(x)dx.\]
As \(\mathbb{E}\left[g_{\epsilon,y_1}(Z(x_1))\right]\xrightarrow{\epsilon\rightarrow 0}F_1(y_1,x_1)\), we have 
\[F_1(y_1,x_1)=c +\int_{0}^{x_1}l(z)dz.\]
We can then consider the integrand in the point \(x_1\), \(l(x_1)\), as the derivative of \(F_1(y_1,x_1)\) w.r.t.\ \(x_1\), i.e.

\[\frac{\partial}{\partial x_1}F_1(y_1,x_1)=-\mathbb{E}\left[S_1(t_{1,T}^*(x_1))\bigg|Z(x_1)=y_1\right]f_{1,x_1}\left(y_1\right),\]
which is finite-valued by assumptions (i) and (iii) for all \(x_1\in(1-\delta,1+\delta)\) and \(y_1\in(q_{{\rm VaR}^{\alpha}_{1,T}}(x_1)-\bar{\delta}/2,q_{{\rm VaR}^{\alpha}_{1,T}}(x_1)+\bar{\delta}/2)\). 
\end{enumerate}
We have now shown that the function \(F_1(y_1,x_1)\) is: 1) continuously differentiable in \(y_1\), and, 2) differentiable in \(x_1\) for \(x_1\in(1-\delta,1+\delta)\) and \(y_1\in(q_{{\rm VaR}^{\alpha}_{1,T}}(x_1)-\bar{\delta}/2,q_{{\rm VaR}^{\alpha}_{1,T}}(x_1)+\bar{\delta}/2)\). 

From the continuity of \(f_{1,x_1}(y_1)\) w.r.t.\ \(y_1\) at \(y_1=q_{{\rm VaR}^{\alpha}_{1,T}}(x_1)\) for all \(x_1\in(1-\delta,1+\delta)\) by assumption (i), we obtain,
\[F_1(q_{{\rm VaR}^{\alpha}_{1,T}}(x_1),x_1)=1-\alpha.\]
By the Implicit Function Theorem~\ref{th_app1} and the differentiabilities derived in items 1) \& 2) above, \(q_{{\rm VaR}^{\alpha}_{1,T}}(x_1)\) is a differentiable function of \(x_1\in(1-\delta,1+\delta)\) with 
\begin{align*}
\frac{\partial q_{{\rm VaR}^{\alpha}_{1,T}}(x_1)}{\partial x_1}=&-\left(f_{1,x_1}\left(q_{{\rm VaR}^{\alpha}_{1,T}}(x_1)\right)\right)^{-1}\frac{\partial}{\partial x_1}F(y_1,x_1)\bigg|_{y_1=q_{{\rm VaR}^{\alpha}_{1,T}}(x_1)}\\
=&\mathbb{E}\left[S_1(t_{1,T}^*(x_1))\big|Z(x_1)=q_{{\rm VaR}^{\alpha}_{1,T}}(x_1)\right].
\end{align*}

The final result follows by setting \(x_1=1\).
\end{proof}


\end{appendices}


\end{document}